\documentclass[12pt,reqno,a4paper]{amsart}
\usepackage{amsmath,amssymb,amsfonts,amscd}
\usepackage[mathscr]{eucal}
\usepackage[T1]{fontenc}
\usepackage{textcomp}
\usepackage{times}

\usepackage{hyperref}

\date{3 (16) June 2012}

\author{Theodore Th. Voronov}
\address{School of Mathematics, University of Manchester,
Oxford Road, Manchester, M13 9PL, United Kingdom}
\email{theodore.voronov@manchester.ac.uk}
\title{$Q$-Manifolds and Mackenzie Theory}

\newtheorem{theorem}{Theorem}
\newtheorem{proposition}{Proposition}

\theoremstyle{definition}
\newtheorem{definition}{Definition}
\newtheorem{example}{Example}[section]
\newtheorem{remark}{Remark}[section]
\newtheorem*{genprinc}{General principle}
\def\co{\colon\thinspace}

\renewcommand{\leq}{\leqslant}

 \DeclareMathOperator{\Vect}{Vect}

\DeclareMathOperator{\id}{id}

\newcommand{\der}[2]{{\frac{\partial {#1}}{\partial {#2}}}}

\newcommand{\Z}{{\mathbb Z_{2}}}
\newcommand{\ZZ}{{\mathbb Z}}
\newcommand{\p}{\partial}
\newcommand{\fun}{C^{\infty}}

\newcommand{\w}{{\boldsymbol{w}}}

\renewcommand{\a}{\alpha}
\renewcommand{\b}{\beta}
\def\e{\varepsilon}

\newcommand{\g}{{\gamma}}
\newcommand{\h}{\eta}
\renewcommand{\t}{\theta}

\newcommand{\lam}{{\lambda}}
\newcommand{\x}{{\xi}}
\def\d{\delta}
\def\t{\theta}

\newcommand{\jtt}{{\tilde \jmath}}
\newcommand{\itt}{{\tilde \imath}}

\newcommand{\ut}{{\tilde u}}
\newcommand{\vt}{{\tilde v}}

\newcommand{\att}{{\tilde \a}}
\newcommand{\btt}{{\tilde \b}}

\newcommand{\be}{{\bar e}}
\newcommand{\bbe}{\Bar {\bar e}}

\begin{document}
\begin{abstract}
Double Lie algebroids were discovered by Kirill Mackenzie from the
study of double Lie groupoids and were defined in terms of rather
complicated conditions making use of duality theory for Lie
algebroids and double vector bundles. In this paper we establish a
simple alternative characterization of double Lie algebroids in a
supermanifold language. Namely, we show that a double Lie algebroid
in Mackenzie's sense is equivalent to a double vector bundle endowed
with a pair of commuting homological vector fields of appropriate
weights. Our approach  helps to simplify and elucidate Mackenzie's
original definition; we show how it fits into a bigger picture of
equivalent structures on `neighbor' double vector bundles. It also
opens ways for extending the theory to   multiple Lie
algebroids, which we introduce here.
\end{abstract}

\maketitle 

\section*{Introduction}

Double Lie algebroids arose in the works on double Lie
groupoids~\cite{mackenzie:secondorder1},
 \cite{mackenzie:secondorder2} and in connection with an analog for
Lie bialgebroids of the classical Drinfeld double of Lie
bialgebras~\cite{mackenzie:doublealg}, 
\cite{mackenzie:drinfeld}.

K.~Mackenzie   put forward the idea that (the analog of) the
\emph{Drinfeld double} for a Lie bialgebroid should be a
\emph{double Lie algebroid} in the Ehresmann sense understood
properly.

Recall that a `double object' in Ehresmann's sense in a  category
$\mathcal{C}$ is an object of the same type $\mathcal{C}$ in the
category $\mathcal{C}$. For example, a group object in the category
of groups (which is, as one can see,  just an Abelian group). This
can  be defined for certain categories where objects have structures
described by diagrams. (As  the above example shows, for some
categories this   may be defined, but be not very interesting.)
`Double Lie groupoids' are therefore groupoid objects in the
category of Lie groupoids. For groupoids, differently from groups,
this leads to a richer structure, not to  a degeneration.

Unlike groupoids, Lie algebroids are not defined diagrammatically.
Because of that, defining double objects for them could not  be done
directly. `Double Lie algebroids' first appeared as the tangent
objects for double Lie groupoids~\cite{mackenzie:secondorder1},
 \cite{mackenzie:secondorder2}. Properties of these tangent
objects were axiomatized later to give the abstract notion. The
definition so obtained~\cite{mackenzie:doublealg} is quite
complicated.  Even stating  the conditions appearing there requires
non-trivial applications of the  duality theory for double vector
bundles and of the relations  between Lie algebroids and Poisson
structures.

Mackenzie justified his definition of double Lie algebroids by
proving the thesis quoted above, i.e., by showing that this
definition is satisfied by a Lie bialgebroid generalization of
Drinfeld's classical double~\cite{mackenzie:doublealg},
\cite{mackenzie:drinfeld}. Hence there was absolutely no doubt that
in~\cite{mackenzie:doublealg} a correct notion was discovered.
However the complexity of the original definition has somewhat
hindered its further applications.

The aim of this paper is to give a simple alternative
characterization of Mackenzie's double Lie algebroids. This is
achieved in the language of $Q$-manifolds, i.e., supermanifolds
endowed with an odd vector field of square zero.

An obvious part of a definition of a double Lie algebroid is, of
course, two structures of ordinary Lie algebroids\,\footnote{In
fact, four Lie algebroid structures, on the four sides of a double
vector bundle; but they can be reduced to the two `main' ones as we
explain in Section~\ref{sec.analysis}.}. The problem was to find
compatibility conditions that have to be satisfied.

In this paper \emph{we analyze Mackenzie's original compatibility
conditions and prove that they are equivalent to the commutativity
of two homological vector fields of suitable weights on a
supermanifold naturally associated with a given double vector
bundle}. This radically simplifies the theory and allows to extend
it immediately to the multiple case, i.e., allows to introduce
$n$-fold Lie algebroids, which before was inaccessible.

The original approach to double Lie algebroids is not made redundant
by our work. In the course of the proof of our main statement, we
show that Mackenzie's `Condition \textrm{III}' (see
Section~\ref{sec.doublealg} below), which pertains to a certain Lie
bialgebroid, actually subsumes his other conditions. With this
simplification, we show further that Mackenzie's original definition
and the construction given in this paper become parts of what we
call a \textbf{big picture}. It is as follows.

Recall that, say, a Lie algebra  $\mathfrak g$ has equivalent
manifestations as a linear Poisson bracket on the coalgebra
$\mathfrak g^*$, as a linear Schouten bracket on the anticoalgebra
$\Pi\mathfrak  g^*$ and as a quadratic homological vector field on
the antialgebra $\Pi\mathfrak g$. See, for
example,~\cite{tv:graded}. The vector spaces $\mathfrak g^*$,
$\Pi\mathfrak g$ and $\Pi\mathfrak  g^*$ are the \textit{neighbors}
of the vector space $\mathfrak g$.

To apply this idea to double Lie algebroids, consider a given double
vector bundle and take all its {neighbors}, that is, the double
vector bundles obtained by dualizations and parity reversions. There
are twelve of them, including the original bundle. Assume the
existence of two ordinary Lie algebroid structures on the original
double vector bundle, with some natural linearity conditions. Then
each of the neighbors acquires a pair of structures such as, e.g., a
Lie algebroid structure and a vector field, a bracket and a vector
field, etc.  We can say that the structure of a `double Lie
algebroid' is manifested in various ways on all of these neighbor
double vector bundles. The crucial thing is to formulate the
compatibility constraint.  To this end, we notice that there are
exactly five cases where the induced structures are defined on the
total space as opposed to the space of sections of a vector bundle.
In each such case  a  compatibility condition comes about naturally
(e.g., that the vector field should be a derivation of the Poisson
bracket). All these natural conditions can be shown to be
equivalent\,---\,by a certain functorial argument. It turns out that
in exactly four of these cases the `natural' compatibility condition
is a reformulation of Mackenzie's Condition \textrm{III}, and in the
remaining case it is precisely our commutativity condition.

The importance of  double Lie groupoids and Lie bialgebroids, and
notions related to them, such as double Lie algebroids, follows,
in particular, from their natural links with Poisson geometry. The
infinitesimal object for a Poisson groupoid is a Lie bialgebroid as
it was shown by Mackenzie and Xu~\cite{mackenzie:bialg}. On the
other hand, if one takes the cotangent bundle $T^*G$ of a
Poisson--Lie group (or a Poisson groupoid) $G$, it is an
`\textit{LA}-groupoid', a notion intermediate between double Lie
groupoids and double Lie algebroids~\cite{mackenzie:secondorder1}.
(Recall that Poisson groupoids incorporate  both Poisson--Lie groups
and symplectic groupoids~\cite{weinstein:symp-groupoids87},
\cite{weinstein:coisotropic88}.) The richness of these notions can
be also seen  in numerous non-obvious structures, isomorphisms and
dualities arising for them, e.g., a  duality theory for double and
triple vector bundles, producing unexpected discrete symmetry
groups~\cite{mackenzie:duality}. We believe that methods developed
in this paper will be particularly useful for all these
applications.

We wish to stress that in our work, supermanifolds provide powerful
tools that we apply to ordinary (``purely even'') objects. Although
we show that everything works also in a `superized' context, this
was not the main goal.

The paper is organized as follows.

In Section~\ref{sec.doublealg} we recall the definition of double
Lie algebroids.

In Section~\ref{sec.background} we recall the description of
(ordinary) Lie algebroids in the language of homological vector
fields, and revise double vector bundles. In particular, we
introduce partial reversions of parity.

In Section~\ref{sec.main} we define double Lie antialgebroids and
give our main statement (Theorem~\ref{thm.main}).

In Section~\ref{sec.analysis} we analyze the three conditions
appearing in the definition of double Lie algebroids and give  a
proof of Theorem~\ref{thm.main}.

In Section~\ref{sec.big} we show how the equivalence of Mackenzie's
notion of double Lie algebroids and our notion of double Lie
antialgebroids is a part of a bigger picture. \textit{Modulo} some
facts established in Section~\ref{sec.analysis}, this provides an
alternative proof of Theorem~\ref{thm.main}.

In Section~\ref{sec.appl} we show the equivalence of Mackenzie's and
Roytenberg's doubles of Lie bialgebroids and discuss an extension of
the whole theory to the multiple   case.

\subsection*{Terminology and notation}

We use the standard language of supermanifolds. The letter $\Pi$
denotes the parity reversion functor, and notation such as
$\Phi^{\Pi}$ is used for linear maps induced on the
\textit{opposite} (parity reversed) objects. Commutators and similar
notions are always understood in the $\Z$-graded sense. A tilde over
an object is used to denote its parity. A \textit{$Q$-manifold}
means a supermanifold endowed with a homological vector field;
likewise, $P$- and \textit{$S$-manifolds} mean those with a Poisson
or Schouten (= odd Poisson) bracket. A \textit{$QS$-manifold}
(resp., a \textit{$QP$-manifold}) means one with $Q$- and $S$-
structures (resp., with $Q$- and $P$-structures) that are compatible
(the vector field is a derivation of the bracket,
cf.~\cite{yvette:exact}). In general, notation and terminology are
close to our paper~\cite{tv:graded}. The space of smooth sections of
a vector bundle $E\to M$ is denoted by $\fun(M,E)$ and the space of
vector fields on a manifold $M$, by $\Vect(M)$. We wish to draw the
reader's special attention to our normally dropping the prefix
\text{`super-'} when this cannot cause confusion and speaking, as a
rule, of `manifolds' meaning supermanifolds, `Lie algebras' meaning
superalgebras, etc.

\section{Double Lie algebroids}\label{sec.doublealg}

Double Lie algebroids were introduced by Mackenzie
in~\cite{mackenzie:doublealg}, see also~\cite{mackenzie:drinfeld},
as the infinitesimal counterparts of double Lie groupoids. The
latter notion is a double object in the sense of Ehresmann, i.e., a
groupoid object in the category of Lie groupoids. Therefore, it has
a natural categorical formulation. Compared to it, the abstract
notion of a double Lie algebroid is rather complicated and
non-obvious. The reason  for this, is that properties of brackets
for Lie algebroids are not expressed by diagrams, so one cannot
approach double objects for them by methods of category theory.
Mackenzie's conditions given below  came about as an abstraction of
the properties of the double Lie algebroid of a double Lie groupoid
discovered in~\cite{mackenzie:secondorder1},
\cite{mackenzie:secondorder2}.

\begin{definition}
\label{def.dlie} A double vector bundle
\begin{equation} \label{eq.dvb}
    \begin{CD} D@>>> B\\
                @VVV  @VVV \\
                A@>>>M
    \end{CD}
\end{equation}
is a \textbf{double Lie algebroid} if all sides are (ordinary) Lie
algebroids and the following \textbf{conditions \em{I}}, \textbf{\em{II}}, and
\textbf{\em{III}} are satisfied:
\begin{description}
  \item[\textbf{\em{I}.}] With  respect to the vertical structures of
  Lie algebroids, $\begin{array}{c} D\\
                \downarrow  \\
                A
    \end{array}$ and
    $\begin{array}{c} B\\
                \downarrow  \\
                M
    \end{array}$, all maps related with the horizontal vector bundle
    structures are Lie algebroid morphisms (more precisely, this includes the projections, the zero sections,
    the fiberwise addition and multiplication by scalars). The same
    holds with vertical/horizontal structures interchanged.
  \item[\textbf{\em{II}.}] The horizontal arrows in the diagram
  \begin{equation*}
    \begin{CD} D@>a>> TB\\
                @VVV  @VVV \\
                A@>a>>TM
    \end{CD}
\end{equation*}
where at the right there is the tangent prolongation of the Lie
algebroid $B\to M$, and the horizontal arrows are the anchors,
define a Lie algebroid morphism. The same holds with
vertical/horizontal structures interchanged.
  \item[\textbf{\em{III}.}] The vertical arrows in the diagram
  \begin{equation*}
    \begin{CD} D^{*A}@>>> K^*\\
                @VVV  @VVV \\
                A@>>>M
    \end{CD}
\end{equation*}
define a Lie algebroid morphism. Here $K$ denotes the core. The same
holds with vertical and horizontal, and $A$ and $B$, interchanged.
The vector bundles in duality $D^{*A}\to K^*$ and $D^{*B}\to K^*$
define a Lie bialgebroid.
\end{description}
\end{definition}

A detailed analysis of these conditions will be given in the next
sections. Now we wish to recall only the following. A \textit{double
vector bundle} such as~\eqref{eq.dvb} is defined by the condition
that all vector bundle structure maps in one direction (horizontal
or vertical) are vector bundle morphisms for the other direction. Its
\textit{core} $K$  is defined as the intersection of the kernels of
the projections $D\to A$ and $D\to B$ considered as vector bundle
morphisms (w.r.t.  the other structure). $K$ is a vector bundle over
$M$. It is a theorem due to Mackenzie~\cite{mackenzie:sympldouble} that taking the two duals of
$D$ considered as a vector bundle either over $A$ or over $B$ leads
to two double vector bundles
\begin{equation*}
    \begin{CD} D^{*A}@>>> K^*\\
                @VVV  @VVV \\
                A@>>>M
    \end{CD} \text{\quad and \quad }
    \begin{CD} D^{*B}@>>> B\\
                @VVV  @VVV \\
                K^*@>>>M
    \end{CD}
\end{equation*}
such that the vector bundles $D^{*A}\to K^*$ and $D^{*B}\to K^*$
over the dual $K^*$ of the core $K$ are\,---\,unexpectedly\,---\,in
a natural duality. All these facts, as well as the notion of the
\textit{tangent prolongation} of a Lie algebroid, can be found
in~\cite[Ch.~9]{mackenzie:book2005}, see also
\cite{mackenzie:duality,mackenzie:double06,mackenzie:double11}. \textit{Lie
bialgebroids} were introduced by Mackenzie and
Xu~\cite{mackenzie:bialg}. Their theory was  advanced by
Y.~Kosmann-Schwarzbach~\cite{yvette:exact}, who in particular gave a
very handy form of the definition, which we use.
See~\cite[Ch.~12]{mackenzie:book2005}.

\section{Lie algebroids and double vector bundles: some background}
\label{sec.background}

In this section we  develop tools that will be later used for an
alternative description of double Lie algebroids (our main goal).

Henceforth we work in the `super' setup, i.e., we consider
supermanifolds and bundles of supermanifolds. However, we
systematically skip the prefix `super-' except when we wish to make
an emphasis.   All the constructions from the previous section carry
over to the super case.

We  use \textit{graded manifolds} as defined in~\cite{tv:graded},
i.e., supermanifolds endowed with an extra $\ZZ$-grading in the
algebras of functions, in general not related with parity. We refer
to such grading as \textit{weight}.

\subsection{Lie algebroids and Lie antialgebroids}\label{sec.algebroids}

Let us recall some known facts concerning  {Lie algebroids}.

It was first  shown by Va\u{\i}ntrob~\cite{vaintrob:algebroids} that Lie
algebroids can be described by homological vector fields. We shall
recall this correspondence using the description given
in~\cite{tv:graded} in the language of derived brackets. As
mentioned, we consider the `superized' version (i.e., `super' Lie
algebroids) by default.

Let $F\to M$ be a vector bundle. The total space $F$ is a graded manifold in a natural way,  the (pullbacks of) functions on the base $M$
having weight $0$ and linear functions on the fibers, weight $1$.
Using weights is very helpful for describing various geometric
objects. For example, vector fields  of weight $-1$ on $F$
correspond to sections of $F$ (or $\Pi F$, see below). Vector fields
of weight $0$ are generators of fiberwise linear transformations.
Vector fields of weight $1$ can be used to generate brackets of
sections. More precisely:  a \textbf{Lie antialgebroid} structure on
$F\to M$, by definition,  is given by a homological field $Q\in
\Vect (F)$ of weight $1$.

\textit{There is a one-to-one correspondence between Lie
antialgebroids and Lie algebroids}, as follows.

Let $\Pi$ denote the parity reversion functor and $F=\Pi E$ for a
vector bundle $E\to M$. Then $F$ is a Lie antialgebroid if and only
if $E$ is a Lie algebroid. The anchor and the bracket  for the
sections of $E$ are given by the following formulas:
\begin{equation}\label{eq.auf}
    a(u)f:=\bigl[[Q,i(u)],f\bigr]
\end{equation}
and
\begin{equation} \label{eq.uv}
    i([u,v]):=(-1)^{\ut}\bigl[[Q,i(u)],i(v)\bigr].
\end{equation}
Here $f\in \fun(M)$, and $u,v\in \fun(M, E)$ are sections. We use
the natural odd injection
\begin{equation}\label{eq.inj}
    i\co \fun(M, E)\to \Vect (\Pi E),
\end{equation}
which sends a section $u\in \fun(M, E)$ to the vector field $i(u)\in
\Vect (\Pi E)$   of weight $-1$. The map $i$ is an odd isomorphism
between the space of sections $\fun(M, E)$ and the subspace
$\Vect_{-1}(\Pi E)\subset \Vect(\Pi E)$ of all vector fields of
weight $-1$. By counting weights, one can see that the LHS's
of~\eqref{eq.auf} and~\eqref{eq.uv} are well-defined. The properties
of the bracket and anchor are deduced from the identity $Q^2=0$ as
is standard in the derived brackets method. Conversely, starting
from a Lie algebroid structure in $E\to M$, one can reconstruct $Q$
on $\Pi E$ with the desired properties.

All these facts can be checked without coordinates; however,
introducing local coordinates makes them particularly transparent.
Let $x^a$ denote local coordinates on the base $M$. We shall use
$u^i$ and $\x^i$ for linear coordinates in the fibers of $E$ and
$F=\Pi E$, respectively. Changes of coordinates have the following
form:
\begin{align*}
    x^a=x^a(x'),\quad
    u^i=u^{i'}T_{i'}{}^{i}(x'),
    \intertext{and}
\x^i=\x^{i'}T_{i'}{}^{i}(x').
\end{align*}
The map $i\co \fun(M, E)\to \Vect (\Pi E)$ has the following
appearance in coordinates{\,}\footnote{Here we denoted the
components $u^i(x)$ of a section $u=u^i(x)e_i$ in the same way as
the fiber coordinates $u^i$. To avoid confusion note that the parity
of a component  $u^i(x)$ is the same as that of the corresponding
variable $u^i$ if the section $u$ is even and it is of the opposite
parity if the section $u$ is odd.}: if $u=u^i(x)e_i$, then
\begin{equation}\label{eq.iu}
    i(u)=(-1)^{\ut}u^i(x)\der{}{\x^i}\,.
\end{equation}
Clearly, the RHS of \eqref{eq.iu} is the general form of a vector
field of weight $-1$ on $F$. A vector field $Q$ of weight $1$ on $F$
in coordinates has the form
\begin{equation*}
    Q=\x^iQ_i^a(x)\,\der{}{x^a}+\frac{1}{2}\,\x^i\x^j Q_{ji}^k
    (x)\,\der{}{\x^k}\,.
\end{equation*}
Equations~\eqref{eq.auf} and \eqref{eq.uv} produce the following
formulas for the anchor:
\begin{equation*}
    a(u)=u^i(x)\,Q_i^a(x)\,\der{}{x^a}\,,
\end{equation*}
and for the brackets:
\begin{equation*}
    [u,v]=\Bigl(u^i Q_i^a\,\p_av^k-(-1)^{\ut(\vt+1)}v^i Q_i^a\,\p_au^k-
    (-1)^{\itt(\vt+1)} u^iv^j Q_{ji}^k\Bigr)e_k\,,
\end{equation*}
where we abbreviated $\p_a=\p/\p x^a$. In particular, for the
elements of the local frame $e_i$ we have
\begin{equation*}
    [e_i,e_j]=(-1)^{\jtt} Q_{ij}^k(x)\,e_k\,.
\end{equation*}

For the record, we shall also include  explicit expressions specifying the  Poisson and Schouten brackets induced on the total spaces of the vector bundles $E^*$ and $\Pi E^*$, respectively. Using the local coordinates $u_i$ and $\x_i$ in the respective fibers (so that the bilinear forms $u^iu_i$ and $u^i\x_i$ are invariant), we obtain
\begin{equation*}
    \{x^a,x^b\}=0\,, \quad \{u_i,x^a\}=Q_i^a(x)\,, \quad \{u_i,u_j\}=(-1)^{\jtt} Q_{ij}^k(x)\,u_k\,
\end{equation*}
for  the Poisson brackets of the coordinates, and
\begin{equation*}
    \{x^a,x^b\}=0\,, \quad \{\x_i,x^a\}=Q_i^a(x)\,, \quad \{\x_i,\x_j\}=(-1)^{\jtt} Q_{ij}^k(x)\,\x_k\,
\end{equation*}
for the Schouten brackets. They appear exactly  the same, but one should not forget  that the Poisson bracket is even and the Schouten bracket is odd. For the sign conventions and other information see~\cite{tv:graded}. For a coordinate-free description of these Lie-Poisson and Lie-Schouten brackets  (without `super') see~\cite[Ch.~7 and Ch.~10]{mackenzie:book2005}.

\medskip
Now let us proceed to double vector bundles.

\subsection{Parity reversions for double vector bundles}\label{sec.parreversion}

For double vector bundles  see~\cite[Ch. 9]{mackenzie:book2005}, \cite{mackenzie:duality},  and the recent paper~\cite{mackenzie:duality2009}. All the necessary notions readily carry over to the supermanifold setup. In particular, this allows to consider parity reversions for double and multiple vector bundles. (Such
operations should be studied together with the duality operations of Mackenzie  theory~\cite{mackenzie:duality} and \cite{mackenzie:duality2009}.)

Let
\begin{equation} \label{eq.dvb2}
    \begin{CD} D@>>> B\\
                @VVV  @VVV \\
                A@>>>M
    \end{CD}
\end{equation}
be a double vector bundle (in the category of supermanifolds). The
manifold $D$ is naturally bi-graded, by weights corresponding to the
two vector bundle structures. We denote these weights by $\w_1$ and
$\w_2$, or by $\w_A$ and $\w_B$, as convenient.

A double vector bundle~\eqref{eq.dvb2} allows \textbf{fiberwise reversions of parity} in both
directions, horizontal and vertical. We denote the corresponding
operations by $\Pi_1$ and $\Pi_2$ (or by $\Pi_A$ and $\Pi_B$ when convenient).  The
vertical reversion of parity $\Pi_1=\Pi_A$ gives
\begin{equation}\label{eq.pa}
    \begin{CD} \Pi_A D@>>> \Pi B\\
                @VVV  @VVV \\
                A@>>>M
    \end{CD}
\end{equation}
which is a new double vector bundle. One can apply the horizontal
reversion of parity to it (applying the vertical reversion again
takes us back), or do it  the other way round. The parity reversions $\Pi_1$ and $\Pi_2$ and their compositions are covariant functors on the category of double vector bundles. (Recall that the morphisms in this category are the morphisms of diagrams~\eqref{eq.dvb2} inducing fiberwise linear maps of all constituent ordinary vector bundles. In the language of graded manifolds, this means simply that a map of   total spaces should preserve both weights  $\w_1$ and $\w_2$.)

\begin{proposition} \label{prop.pp} The operations $\Pi_1$ and $\Pi_2$ commute:
\begin{equation*}
    \Pi_1\Pi_2=\Pi_2\Pi_1\,,
\end{equation*}
or, more precisely, there is a canonical isomorphism of functors
\begin{equation*}
    I_{12}\co  \Pi_2\Pi_1\to \Pi_1\Pi_2\,.
\end{equation*}
In  greater detail,  for each
double vector bundle given by~\eqref{eq.dvb2},
there is an isomorphism of double vector bundles
\begin{equation} \label{eq.pi1pi2}
    \begin{CD} \Pi_B\Pi_A D@>>> \Pi B\\
                @VVV  @VVV \\
                \Pi A@>>>M
    \end{CD}
    \text{\quad $\stackrel{\quad I_{12}\quad}{\longrightarrow}$ \quad}
    \begin{CD} \Pi_A\Pi_B D@>>> \Pi B\\
                @VVV  @VVV \\
                \Pi A@>>>M
    \end{CD}
\end{equation}
(where we used the more suggestive notation $\Pi_A=\Pi_1$ and
$\Pi_B=\Pi_2$) commuting with the morphisms  induced by all  double vector bundles morphisms    $\Phi\co D_1\to D_2\,$.
\end{proposition}

We can denote the common value of the ultimate total spaces
by $\Pi^2 D$ (up to a natural isomorphism). If we need a particular choice of the double vector bundle, we can agree {for concreteness} that $\Pi^2 D=\Pi_2\Pi_1 D$\,.
We  call $\Pi^2 D$ the \textit{complete parity reversion} of $D$.

\begin{remark} Everything extends immediately to the  $n$-fold
case, with $\Pi^n D$ being the \textit{complete parity reversion} of
the ultimate total space $D$ of an $n$-fold vector bundle. There are
partial parity reversion operations $\Pi_{r}$ such that
$\Pi_{r}\Pi_{s}=\Pi_{s}\Pi_{r}$ (equality means natural isomorphism) and $\Pi^n D=\Pi_n\cdot \ldots
\cdot\Pi_1 D$.   We   consider multiple vector bundles in
Section~\ref{sec.appl}.
\end{remark}

We shall explain now the constructions and give a proof of Proposition~\ref{prop.pp} using local coordinates.
The coordinate language is particularly handy for visualizing double
(and multiple) vector bundles; it is as follows. Consider a double
vector bundle given by~\eqref{eq.dvb2}. As above, denote local
coordinates on $M$ by $x^a$. Let $u^i$ and $w^{\a}$ be linear
coordinates on the fibers of $A\to M$ and $B\to M$, respectively. On
$D$ we have coordinates $x^a, u^i, w^{\a}, z^{\mu}$ so that $u^i,
z^{\mu}$ are linear fiber coordinates for $D\to B$ and $w^{\a},
z^{\mu}$, for $D\to A$. Coordinate changes have the form:
\begin{align}
    x^a&=x^a(x'),\\
    u^i&=u^{i'}T_{i'}{}^{i}(x'), \label{eq.lawforu}\\
    w^{\a}&=w^{\a'}T_{\a'}{}^{\a}(x'),\\
    z^{\mu}&=z^{\mu'}T_{\mu'}{}^{\mu}(x')+u^{i'}w^{\a'}T_{\a'
    i'}{}^{\mu}(x')\,. \label{eq.lawforz}
\end{align}
(Note that the submanifold  specified by the equations $u^i=0$ and $w^{\a}=0$ is the core $K$ of the double vector bundle   $D$ and the variables $z^{\mu}$ restricted to $K$ become  fiber coordinates for the vector bundle $K\to M$.)
This is a convenient coordinate description of a double vector
bundle structure.
(The reader who lacks a taste for coordinates may
translate everything into a language of local trivializations.)
In
particular, for two weights we have $\w_1=\# u+\# z$ and $\w_2=\#
w+\# z$, where $\#$ denotes the degree in the respective variable.
Everything extends directly to multiple vector bundles, see Section~\ref{sec.appl}.

Partial parity reversions can now be described as follows.

For the double vector bundle given by~\eqref{eq.pa},
we have $x^a, u^i, \h^{\a}, \t^{\mu}$ as coordinates on $\Pi_AD$, so
that $\widetilde{\h^{\a}}=\widetilde{w^{\a}}+1=\widetilde{\a}+1$,
$\widetilde{\t^{\mu}}=\widetilde{z^{\mu}}+1=\widetilde{\mu}+1$, and
the changes of coordinates are
\begin{align*}
    u^i&=u^{i'}T_{i'}{}^{i},\\
    \h^{\a}&=\h^{\a'}T_{\a'}{}^{\a},\\
    \t^{\mu}&=\t^{\mu'}T_{\mu'}{}^{\mu}+(-1)^{\itt'}u^{i'}\h^{\a'}T_{\a'
    i'}{}^{\mu}\,,
\end{align*}
where we suppress coordinates on $M$. Here $\h^{\a}, \t^{\mu}$ are
fiber coordinates for $\Pi_A D\to A$ and $\h^{\a}$ are fiber
coordinates for $\Pi B\to M$. (Here and below, switching from Latin to Greek letters and back is meant to remind about changing of parity.)

Similarly, for
\begin{equation}\label{eq.pb}
    \begin{CD} \Pi_B D@>>>  B\\
                @VVV  @VVV \\
                \Pi A@>>>M
    \end{CD}
\end{equation}
we have $x^a, \x^i, w^{\a}, \e^{\mu}$ as coordinates on $\Pi_BD$,
where $\widetilde{\x^{i}}=\widetilde{u^{i}}+1={\itt}+1$ and
$\widetilde{\e^{\mu}}=\widetilde{z^{\mu}}+1=\widetilde{\mu}+1$, with
the changes of coordinates
\begin{align*}
    \x^i&=\x^{i'}T_{i'}{}^{i},\\
    w^{\a}&=w^{\a'}T_{\a'}{}^{\a},\\
    \e^{\mu}&=\e^{\mu'}T_{\mu'}{}^{\mu}+\x^{i'}w^{\a'}T_{\a'
    i'}{}^{\mu}\,.
\end{align*}

By applying to the bundles $\Pi_A D$ and $\Pi_B D$ the parity reversions in the other directions, we
obtain, respectively, the double vector bundle
\begin{equation} \label{eq.pibpia}
    \begin{CD} \Pi_B\Pi_A D@>>> \Pi B\\
                @VVV  @VVV \\
                \Pi A@>>>M
    \end{CD}
\end{equation}
with coordinates  $x^a, \x^i, \h^{\a}, t^{\mu}$ on $\Pi_B\Pi_A D$
and the transformation law
\begin{align*}
    \x^i&=\x^{i'}T_{i'}{}^{i},\\
    \h^{\a}&=\h^{\a'}T_{\a'}{}^{\a},\\
    t^{\mu}&=t^{\mu'}T_{\mu'}{}^{\mu}+(-1)^{\itt'}\x^{i'}\h^{\a'}T_{\a'
    i'}{}^{\mu}\,,
\end{align*} and the double vector bundle
\begin{equation} \label{eq.piapib}
    \begin{CD} \Pi_A\Pi_B D@>>> \Pi B\\
                @VVV  @VVV \\
                \Pi A@>>>M
    \end{CD}
\end{equation}
with coordinates $x^a, \x^i, \h^{\a}, s^{\mu}$ on $\Pi_A\Pi_B D$ and
the transformation law
\begin{align*}
    \x^i&=\x^{i'}T_{i'}{}^{i},\\
    \h^{\a}&=\h^{\a'}T_{\a'}{}^{\a},\\
    s^{\mu}&=s^{\mu'}T_{\mu'}{}^{\mu}+(-1)^{\itt'+1}\x^{i'}\h^{\a'}T_{\a'
    i'}{}^{\mu}\,.
\end{align*}
{
The explanation of the sign factors appearing before the second terms in the transformation laws for $\t^{\mu}$, $\e^{\mu}$, $t^{\mu}$ and $s^{\mu}$ above is as follows. To calculate the transformation for the fiber coordinates in the parity-reversed vector bundle, we have to express the transformation law in the original bundle so that   its fiber coordinates all stand at the left and then to replace them by the variables of the opposite parity (see, e.g.,~\cite{tv:graded}). So, for example, for $\Pi_AD\to A$, we take $z^{\mu}=z^{\mu'}T_{\mu'}{}^{\mu}(x')+u^{i'}w^{\a'}T_{\a'i'}{}^{\mu}(x')$ and re-write it as $z^{\mu}=z^{\mu'}T_{\mu'}{}^{\mu}(x')+(-1)^{\itt'\att'}w^{\a'}u^{i'}T_{\a'i'}{}^{\mu}(x')$. By replacing
$w,z$ by $\h,\t$, %
%
resp., we arrive at $\t^{\mu}=\t^{\mu'}T_{\mu'}{}^{\mu}(x')+(-1)^{\itt'\att'}\h^{\a'}u^{i'}T_{\a'i'}{}^{\mu}(x')$, which, after re-arranging back to the initial order, gives $\t^{\mu}=\t^{\mu'}T_{\mu'}{}^{\mu}(x')+(-1)^{\itt'}u^{i'}\h^{\a'}T_{\a'i'}{}^{\mu}(x')$ because the parities of $w^{\a'}$ and $\h^{\a'}$ are opposite. A short-cut allowing to visualize this easily is to treat the passage to the linear coordinates of reversed  parity as a formal multiplication (from the left) by an odd symbol $\Pi$ satisfying the usual commutation relations with whatever coefficients (so that in our example $\h$ and $\t$ are regarded as $\Pi w$ and $\Pi z$).
}

Note that the transformation laws for the $\x$- and $\h$-coordinates  are the same on $\Pi_B\Pi_A D$ and $\Pi_A\Pi_B D$ (that is why we used for them the same letters), and the
transformation laws for the $t$-coordinates on $\Pi_B\Pi_A D$ and the $s$-coordinates on $\Pi_A\Pi_B D$ differ only by
a sign. We define a transformation of double vector bundles
\begin{equation*}
    I_{12}\co \Pi_B\Pi_A D\to \Pi_A\Pi_B D
\end{equation*}
by the formulas $I_{12}^*(\x^i)=\x^i$, $I_{12}^*(\h^{\a})=\h^{\a}$, and $I_{12}^*(s^{\mu})=-t^{\mu}$\,. Clearly, it is well-defined and is an isomorphism. If we have an arbitrary morphism of double vector bundles $\Phi\co D_1\to D_2$, it induces morphisms  $\Pi_1\Pi_2 D_1\to \Pi_1\Pi_2 D_2$ and $\Pi_2\Pi_1 D_1\to \Pi_2\Pi_1 D_2$ (here the indices of $\Pi_i$ indicate the ``directions'' of partial parity reversions and the indices of $D_1$, $D_2$ are   the labels of double vector bundles in consideration). In coordinates they are as follows. For the original bundles:
 \begin{align*}
    \Phi^*(x_2^{a_2})&=x_2^{a_2}(x_1),\\
    \Phi^*(u_2^{i_2})&=u_1^{i_1}\Phi_{i_1}^{i_2}(x_1), \\
    \Phi^*(w_2^{\a_2})&=w_1^{\a_1}\Phi_{\a_1}^{\a_2}(x_1),\\
    \Phi^*(z_2^{\mu_2})&=z_1^{\mu_1}\Phi_{\mu_1}^{\mu_2}(x_1)+u_1^{i_1}w_1^{\a_1}\Phi_{\a_1
    i_1}^{\mu_2}(x_1)\,.
\end{align*}
For the induced morphisms of the parity reversed bundles we obtain from here, for $\Pi_1\Pi_2D_1\to \Pi_1\Pi_2D_2\,$:
\begin{align*}
    \Phi^*(\x_2^{i_2})&=\x_1^{i_1}\Phi_{i_1}^{i_2}(x_1), \\
    \Phi^*(\h_2^{\a_2})&=\h_1^{\a_1}\Phi_{\a_1}^{\a_2}(x_1),\\
    \Phi^*(s_2^{\mu_2})&=s_1^{\mu_1}\Phi_{\mu_1}^{\mu_2}(x_1)+(-1)^{\itt_1+1}\x_1^{i_1}\h_1^{\a_1}\Phi_{\a_1
    i_1}^{\mu_2}(x_1)\,,
\end{align*}
and for $\Pi_2\Pi_1D_1\to \Pi_2\Pi_1D_2\,$:
\begin{align*}
    \Phi^*(\x_2^{i_2})&=\x_1^{i_1}\Phi_{i_1}^{i_2}(x_1), \\
    \Phi^*(\h_2^{\a_2})&=\h_1^{\a_1}\Phi_{\a_1}^{\a_2}(x_1),\\
    \Phi^*(t_2^{\mu_2})&=t_1^{\mu_1}\Phi_{\mu_1}^{\mu_2}(x_1)+(-1)^{\itt_1}\x_1^{i_1}\h_1^{\a_1}\Phi_{\a_1
    i_1}^{\mu_2}(x_1)\,.
\end{align*}
(Strictly speaking, we have to use the notations such as $\Phi^{\Pi_1\Pi_2}$ for the induced morphisms, but we abbreviate them  just to $\Phi$. The explanation for the signs is as above.) The transformation $I_{12}$, for each of the double vector bundles, maps $(\x,\h,t)$ to $(\x,\h,s)$ with $s=-t$. Therefore the diagram
\begin{equation*}
    \begin{CD} \Pi_2\Pi_1 D_1@>{\Phi}>> \Pi_2\Pi_1 D_2\\
                @V{I_{12}}VV  @VV{I_{12}}V \\
                \Pi_1\Pi_2 D_1@>>{\Phi}> \Pi_1\Pi_2 D_2
    \end{CD}
\end{equation*}
is commutative; i.e., we see that the transformation $I_{12}$ commutes with the induced morphisms or is ``natural'' in the categorical sense (an isomorphism of functors). This completes a proof of
Proposition~\ref{prop.pp}. \qed

\begin{remark} Our choice of   $I_{12}$ is not the only possible. One could prefer to change the sign of the ``side'' coordinates such as $\x^i$ or $\h^{\a}$. Our choice is to keep the   transformation $I_{12}$ identical on the sides and inducing $-\id$ on the core bundle. {This is the isomorphism used in the main statement below.}
\end{remark}

\section{Main statement}\label{sec.main}

In this section we give our main statement, which is a
characterization of Mackenzie's double Lie algebroids in terms of
graded $Q$-manifolds. Proofs will be given in Sections
\ref{sec.analysis} and \ref{sec.big}.

\begin{definition} \label{def.dalie} A double vector bundle
\begin{equation} \label{eq.dvb3}
    \begin{CD} H@>>> G\\
                @VVV  @VVV \\
                F@>>>M
    \end{CD}
\end{equation}
is a \textbf{double Lie antialgebroid} if the manifold $H$ is
endowed with two homological vector fields $Q_1$ and $Q_2$  of
weights $(1,0)$ and $(0,1)$, respectively, such that
\begin{equation}\label{eq.q1q2}
    [Q_1,Q_2]=0.
\end{equation}
\end{definition}

{
\begin{remark} In our earlier text~\cite{tv:mack}, it was stated erroneously that,
by taking $Q_1+Q_2$, a double antialgebroid structure can be further reduced to a single
homological field $Q$ of total weight $1$. However,   such a field $Q$   decomposed according
to the weights $\w_1$ and $\w_2$ would include terms of
weights $(2,-1)$ and $(-1,2)$    besides the `correct' terms
of weights $(1,0)$ and $(0,1)$. So we indeed need two fields $Q_1$ and $Q_2$   and cannot get away with a single condition on total weight. 
\end{remark}
}

\begin{remark} \label{rem.mantialg}
An extension to \textit{multiple Lie antialgebroids} is immediate.
An \textbf{$n$-fold Lie antialgebroid} is an $n$-fold vector bundle,
which therefore gives rise to an $n$-graded structure on its
(ultimate) total space, endowed with $n$ commuting homological
fields $Q_{r}$, $r=1,\ldots,n$, on the  total space of weights
$(0,\ldots, 1, \ldots, 0)$, respectively. Here $1$ stands at the
$r$-th place, all other weights being zero. We shall elaborate this
in Section~\ref{sec.appl}.
\end{remark}

\begin{example} Let $(E, E^*)$ be a Lie
bialgebroid with base $M$. Then
\begin{equation} \label{eq.dimadouble}
    \begin{CD} T^*\Pi E=T^*\Pi E^*@>>> \Pi E^*\\
                @VVV  @VVV \\
                \Pi E@>>>M
    \end{CD}
\end{equation}
is a double vector bundle, as one can check. This is a superization
of Mackenzie~\cite{mackenzie:doublealg}. The natural diffeomorphism
$T^*\Pi E= T^*\Pi E^*$ in the upper-left corner
of~\eqref{eq.dimadouble} is a superized version~\cite{tv:graded} of
a theorem of Mackenzie and Xu~\cite{mackenzie:bialg} extending a
statement of Tulczyjew~\cite{tulczyjew:1977}. The Lie algebroid

structures in $E$ and $E^*$ give rise to homological fields on $\Pi
E$ and $\Pi E^*$, respectively.  These two vector fields correspond
to two functions (`linear Hamiltonians') on $T^*\Pi E=T^*\Pi E^*$,
of weights $(1,0)$ and $(0,1)$. That $(E,E^*)$ is a  bialgebroid  is
equivalent to the commutativity of these Hamiltonians (due to
Roytenberg~\cite{roytenberg:thesis}, see also \cite{tv:graded}).
Therefore the corresponding Hamiltonian vector fields make the
double vector bundle~\eqref{eq.dimadouble}  a double Lie
antialgebroid. We shall come back to this example in
Section~\ref{sec.appl}.
\end{example}

\begin{theorem}[A Characterization of Double Lie
Algebroids]\label{thm.main} A double Lie algebroid structure  in a
double vector bundle such as~\eqref{eq.dvb2} is equivalent to a
double Lie antialgebroid structure in the complete parity reversion
double vector bundle
\begin{equation} \label{eq.pikvadrat}
    \begin{CD} \Pi^2 D@>>> \Pi B\\
                @VVV  @VVV \\
                \Pi A@>>>M
    \end{CD}
\end{equation}
i.e., to a pair of commuting homological fields $Q_1$ and $Q_2$  of
weights $(1,0)$ and $(0,1)$ on $\Pi^2 D$.
\end{theorem}

To appreciate the statement one may compare the three conditions in
Definition~\ref{def.dlie} with   simple equation~\eqref{eq.q1q2}.

{
\begin{remark}
We may consider $\Pi^2 D$ as $\Pi_B\Pi_A D$ or as $\Pi_A\Pi_BD$.  The correspondence between the Lie algebroids and the Lie antialgebroids given by Theorem~\ref{thm.main} does not depend on this choice. However, it is important that   the isomorphism $I_{12}$  defined in Proposition~\ref{prop.pp} is used  for an identification of $\Pi_A\Pi_BD$ and $\Pi_B\Pi_AD$.
\end{remark}
}

Let us show how homological vector fields on $\Pi^2 D$ as in the Theorem generate  Lie
algebroid structures on all sides of~\eqref{eq.dvb2}. As in our
discussion of a single Lie algebroid above, everything can be
formulated in a coordinate-free setting. However, using coordinates
sheds some extra light and gets to the formulas quicker.

For the sake of concreteness,  consider $\Pi^2 D=\Pi_B\Pi_A D$
with natural coordinates $x^a, \x^i, \h^{\a}, t^{\mu}$ thereon.
Consider a homological   vector field $Q_1\in \Vect (\Pi^2 D)$ of
weight $(1,0)$ and a homological   vector field $Q_2\in \Vect (\Pi^2
D)$ of weight $(0,1)$. They have the general forms
\begin{multline}\label{eq.fieldqa}
    Q_1=\x^iQ_i^a\der{}{x^a}+\frac{1}{2}\x^i\x^jQ_{ji}^k\der{}{\x^k} \\
    +\left(\x^i\h^{\a}Q_{\a i}^{\b}+t^{\mu}Q_{\mu}^{\b}\right)\der{}{\h^{\b}}
    +\left(\frac{1}{2}\x^i\x^j\h^{\a}Q_{\a ji}^{\lambda}+\x^it^{\mu}Q_{\mu i}^{\lambda}\right)
    \der{}{t^{\lambda}}\,,
\end{multline}
and
\begin{multline}\label{eq.fieldqb}
    Q_2=\h^{\a}Q_{\a}^a\der{}{x^a}+
    \left(\h^{\a}\x^iQ_{i\a}^{j}+t^{\mu}Q_{\mu}^{j}\right)\der{}{\x^{j}}\\
    + \frac{1}{2}\h^{\a}\h^{\b}Q_{\b\a}^{\g}\der{}{\h^{\g}}+
    \left(\frac{1}{2}\h^{\a}\h^{\b}\x^{i}Q_{i \b\a}^{\lambda}+\h^{\a}t^{\mu}Q_{\mu\a}^{\lambda}\right)
    \der{}{t^{\lambda}}\,,
\end{multline}
dictated by their respective weights. All coefficients here are functions of $x^a$. Now,
\textit{due to the fact that the vector field $Q_1$ has weight $0$
w.r.t. the vertical fiber coordinates, it admits partial parity
reversion in this direction}, giving a vector field on $\Pi_B D$ with fiber coordinates $\x^i,w^{\a},\e^{\mu}\,$:
\begin{multline}\label{eq.fieldqapi}
    Q_1^{\Pi}=\x^iQ_i^a\der{}{x^a}+\left((-1)^{\itt}\x^i w^{\a}Q_{\a i}^{\b}+\e^{\mu}Q_{\mu}^{\b}\right)\der{}{w^{\b}}  \\ +\frac{1}{2}\x^i\x^jQ_{ji}^k\der{}{\x^k}
    +\left(\frac{1}{2}(-1)^{\itt+\jtt}\x^i\x^j w^{\a}Q_{\a ji}^{\lambda}+(-1)^{\itt}\x^i\e^{\mu}Q_{\mu i}^{\lambda}\right)
    \der{}{\e^{\lambda}}
\end{multline}
(we have regrouped terms here).
Similarly, $Q_2$ \textit{admits the vertical parity reversion},
which gives a vector field on   $\Pi_A D$ with fiber coordinates $u^i,\h^{\a},\t^{\mu}\,$:
\begin{multline}\label{eq.fieldqbpi}
    Q_2^{\Pi}=\h^{\a}Q_{\a}^a\der{}{x^a}+
    \left((-1)^{\att}\h^{\a}u^iQ_{i\a}^{j}-\t^{\mu}Q_{\mu}^{j}\right)\der{}{u^{j}}\\
    + \frac{1}{2}\h^{\a}\h^{\b}Q_{\b\a}^{\g}\der{}{\h^{\g}}+
    \left(-(-1)^{\att+\btt}\frac{1}{2}\h^{\a}\h^{\b}u^{i}Q_{i \b\a}^{\lambda}+(-1)^{\att}\h^{\a}\t^{\mu}Q_{\mu\a}^{\lambda}\right)
    \der{}{\t^{\lambda}}\,.
\end{multline}
(To find these vector fields, one has to write down the infinitesimal flows of the original fields, which will be linear in the corresponding directions,  and apply the parity reversions to them.)

Both $Q_1^{\Pi}$ and $Q_2^{\Pi}$ are again homological fields. They
define Lie antialgebroid structures on the vector bundles $\Pi_BD\to
B$ and $\Pi_A D\to A$, which correspond to Lie algebroid structures
on $D\to B$ and $D\to A\,$, respectively.

The restrictions of $Q_1^{\Pi}$ and $Q_2^{\Pi}$ on $\Pi A$ and $\Pi
B$, respectively, treated as submanifolds (zero sections) in
$\Pi_BD$ and $\Pi_A D$ are tangent to these submanifolds and define
the homological vector fields
\begin{equation}\label{eq.fieldqapi2}
    Q_1^{(0)}=\x^iQ_i^a\der{}{x^a}+\frac{1}{2}\x^i\x^jQ_{ji}^k\der{}{\x^k}
\end{equation}
on $\Pi A$ and
\begin{equation}\label{eq.fieldqbpi2}
    Q_2^{(0)}=\h^{\a}Q_{\a}^a\der{}{x^a}
    + \frac{1}{2}\h^{\a}\h^{\b}Q_{\b\a}^{\g}\der{}{\h^{\g}}
\end{equation}
on $\Pi B$. This gives   Lie algebroid structures on $A\to M$ and
$B\to M$.

\section{Analysis of Mackenzie's conditions}\label{sec.analysis}

To prove Theorem~\ref{thm.main}, we shall go in the direction
opposite to that in the previous section. (For this reason, we shall change our notation for the homological vector fields slightly, as the reader will see.) Consider a double vector
bundle given by~\eqref{eq.dvb}. \emph{Assume that all four sides are Lie
algebroids and describe them by homological vector fields.} Then we
study conditions \textrm{I}, \textrm{II} and \textrm{III} of Definition~\ref{def.dlie} and see
what they mean in terms of these fields.

Recall the notion of a \textit{Lie algebroid morphism}. It is
non-obvious for  different bases.
See~\cite[$\S$4.3]{mackenzie:book2005} for the definition. Instead
of it, we shall use the  following  statement:

\begin{proposition}[Va\u{\i}ntrob]\label{prop.lamorphism} Suppose $E_1\to M_1$ and $E_2\to M_2$ are
Lie algebroids defined by homological vector fields $Q_1\in \Vect
\Pi E_1$ and $Q_2\in \Vect \Pi E_2$. A vector bundle map
\begin{equation*}
    \begin{CD} E_1 @>{\Phi}>> E_2 \\
                @VVV    @VVV\\
                M_1 @>{\varphi}>> M_2
    \end{CD}
\end{equation*}
is a Lie algebroid morphism
if and only if  the vector  fields $Q_1$ and $Q_2$ are
$\Phi^{\Pi}$-related, where $$\Phi^{\Pi}\co \Pi E_1\to \Pi E_2$$ is
the  induced map of the opposite vector bundles.
\end{proposition}

This statement first appeared, without proof, in~\cite{vaintrob:algebroids}. It is equivalent to saying that
$\Phi^{\Pi}\co \Pi E_1\to \Pi E_2$ is a morphism of Lie
antialgebroids. This condition is much easier to handle than the original
definition of Lie algebroid morphisms.
Recall that vector fields on two (super)manifolds are $F$-\textit{related} (or are \emph{intertwined by $F$}),
for a smooth map $F$, if on smooth functions $F^*\circ Y=X\circ F^*$. In terms of the local flows $g_t$ and $h_t$ generated by $Y$ and $X$, this means that the map $F$ intertwines   the  flows:
$g_t\circ F=F\circ h_t$.

Let us introduce some  notation. We are given  Lie algebroid structures
in the vector bundles $D\to B$ and $A\to M$, and in $D\to A$ and
$B\to M$. Hence we have four homological vector fields:  $Q_{DB}\in \Vect
(\Pi_BD)$,  $Q_{AM}\in \Vect (\Pi A)$,  $Q_{DA}\in \Vect (\Pi_AD)$
and $Q_{BM}\in \Vect (\Pi B)$. In the notation from Section~\ref{sec.background}, the   vector bundle $\Pi_BD\to B$ has fiber coordinates $\x^i, \e^{\mu}$. The vector field $Q_{DB}$ should have
weight $1$ in these variables. In the same way, the   vector bundle $\Pi A \to M$ has fiber coordinates $\x^i$ and the vector field   $Q_{AM}$ is of weight $1$ in
$\x^i$. Similarly, $Q_{DA}$ has weight $1$ in the variables
$\h^{\a},\t^{\mu}$ and $Q_{BM}$, in the variables $\h^{\a}$ (the fiber coordinates for the vector bundles $\Pi_AD\to A$ and $\Pi B\to M$, respectively).

\smallskip
We shall now study conditions \textrm{I}, \textrm{II} and \textrm{III} one by one.

\subsection{Condition \em{I}}
{Condition \textrm{I}} is the easiest for  analysis.

Consider for concreteness the horizontal algebroid structures.
Condition \textrm{I} requires that all vertical structure maps: bundle
projections, zero sections, fiberwise addition and multiplication by
scalars, give morphisms of Lie algebroids. We have the following
diagrams to analyze:

\begin{equation*}
    \begin{CD} D @>>> B \\
                @V{p}VV    @VV{p}V\\
                A @>>> M
    \end{CD}
\text{\quad and \quad }
\begin{CD} D @>>> B \\
                @A{i}AA    @AA{i}A\\
                A @>>> M
    \end{CD}
\end{equation*}
and
\begin{equation*}
    \begin{CD} D \times_A D @>>> B\times_M B \\
                @V{+_A}VV    @VV{+_M}V\\
                D @>>> B
    \end{CD}
\text{\quad and \quad }
\begin{CD} D @>>> B \\
                @V{t_A}VV    @VV{t_M}V\\
                D @>>> B
    \end{CD}
\end{equation*}

In the language of homological vector fields, we see that the flows
generated by the vector fields $Q_{DB}$ and $Q_{AM}$ should commute
with all the vertical structure maps above, more precisely, with the
maps induced on the total spaces of the parity reversed horizontal
vector bundles. Commuting with the projection $p$ means that (the flow
of) the vector field $Q_{DB}$ acts fiberwise on the total space of
$\Pi_BD\to \Pi A$ and induces  on $\Pi A$ (the flow of) the vector
field $Q_{AM}$. Hence  $Q_{AM}$ is completely determined by
$Q_{DB}$. Consider the action of the flow of $Q_{DB}$ on the fibers
of $\Pi_BD\to \Pi A$. Commutativity with   fiberwise multiplication
by scalars, $t_{A}\co \Pi_BD \to \Pi_BD$, and addition, ${+}_{A}\co
\Pi_BD \times_{\Pi A} \Pi_BD \to \Pi_BD$, means that the flow of
$Q_{DB}$ is fiberwise linear (over $\Pi A$). This is equivalent to
the vector field $Q_{DB}$ having weight $0$ w.r.t. fiber coordinates
on $\Pi_BD\to \Pi A$. Commutativity with the zero section $i\co \Pi A\to
\Pi_BD$ then comes about automatically.

We may summarize: if the horizontal Lie algebroid structures are
described by homological vector fields  $Q_{DB}$ and $Q_{AM}$, then
\textit{Condition \em{I} of Definition~\ref{def.dlie} is equivalent to
$Q_{DB}$ having vertical weight $0$ (its horizontal weight being $1$)
and  $Q_{AM}$ being the restriction of $Q_{DB}$ to the base $\Pi
A\subset \Pi_BD$}.

In the same way, for  the vector fields $Q_{DA} \in \Vect
(\Pi_A D)$ and $Q_{BM} \in \Vect (\Pi B)$ describing vertical Lie
algebroid structures, we obtain that \textit{ $Q_{DA}$ should have weight
$(0,1)$ on $\Pi_AD$ and $Q_{BM}$  be its restriction to $\Pi B$}.

In coordinates we arrive at the following general expressions for $Q_{DB}$,  $Q_{AM}$,  $Q_{DA}$
and $Q_{BM}$ dictated by their weights.  For the vector fields describing the horizontal Lie algebroid structures:
\begin{multline}\label{eq.qdb}
    Q_{DB}=\x^iQ_i^a\der{}{x^a}+\left(\x^i w^{\a}Q_{\a i}^{\b}+\e^{\mu}Q_{\mu}^{\b}\right)\der{}{w^{\b}} \\
    +\frac{1}{2}\x^i\x^jQ_{ji}^k\der{}{\x^k}
    +\left(\frac{1}{2}\x^i\x^j w^{\a}Q_{\a ji}^{\lambda}+\x^i\e^{\mu}Q_{\mu i}^{\lambda}\right)
    \der{}{\e^{\lambda}}\,,
\end{multline}
a vector field of weight $(1,0)$ on $\Pi_B D\,$, and
\begin{equation}\label{eq.qam}
    Q_{AM}=\x^iQ_i^a\der{}{x^a}+\frac{1}{2}\x^i\x^jQ_{ji}^k\der{}{\x^k}\,,
\end{equation}
a vector field  of weight $1$ on $\Pi A$, which is a  restriction of $Q_{DB}$
on $\Pi A\,$. Similarly, for  the vector fields  describing the vertical Lie
algebroid structures, we have the following general form:
\begin{multline}\label{eq.qda}
    Q_{DA}=\h^{\a}Q_{\a}^a\der{}{x^a}+
    \left(\h^{\a}u^iQ_{i\a}^{j}+\t^{\mu}Q_{\mu}^{j}\right)\der{}{u^{j}}\\
    + \frac{1}{2}\h^{\a}\h^{\b}Q_{\b\a}^{\g}\der{}{\h^{\g}}+
    \left(\frac{1}{2}\h^{\a}\h^{\b}u^{i}Q_{i \b\a}^{\lambda}+\h^{\a}\t^{\mu}Q_{\mu\a}^{\lambda}\right)
    \der{}{\t^{\lambda}}\,,
\end{multline}
and
\begin{equation}\label{eq.qbm}
    Q_{BM}=\h^{\a}Q_{\a}^a\der{}{x^a}
    + \frac{1}{2}\h^{\a}\h^{\b}Q_{\b\a}^{\g}\der{}{\h^{\g}}\,.
\end{equation}

\textbf{Warning:} formulas~\eqref{eq.qdb}--\eqref{eq.qbm} so obtained are similar to formulas~\eqref{eq.fieldqapi}--\eqref{eq.fieldqbpi2}, but there is a  different sign convention in the notation for the coefficients. We shall stick to the conventions as in~\eqref{eq.qdb}--\eqref{eq.qbm} from now on. This   choice of signs  is explained by the wish to have simpler expressions for   $Q_{DB}$ and $Q_{DA}$ arising  as   primary objects.

For the homological vector fields given by~\eqref{eq.qdb}--\eqref{eq.qbm},  we need to
deduce further constraints  corresponding to Mackenzie's Conditions \textrm{II} and \textrm{III}.
\begin{remark}
It is worth re-iterating that
the  vector fields $Q_{DB}$ and $Q_{DA}$ determining the horizontal
and vertical Lie algebroid structures for $D$   are defined on
different supermanifolds $\Pi_BD$ and $\Pi_AD$.
The crucial fact however is that each of them has weight zero in the second
direction. This will  allow  us to apply an additional parity
reversion and arrive finally at  a pair of vector fields $Q_1$  and $Q_2$ defined on a
common domain $\Pi^2 D$  (by using the theory developed in \S\ref{sec.parreversion}).
\end{remark}

\subsection{Condition {\em II}}

Consider Condition \textrm{II} of Definition~\ref{def.dlie}. For the diagram
\begin{equation} \label{eq.datotbtm1}
    \begin{CD} D@>a>> TB\\
                @VVV  @VVV \\
                A@>a>>TM
    \end{CD}
\end{equation}
which is supposed to give a Lie algebroid morphism
\begin{equation}\label{eq.datotbtm}
\begin{CD}
    \begin{CD} D\\
                @VVV \\
                A
    \end{CD}
    @>>>
    \begin{CD}  TB\\
                 @VVV \\
                TM
    \end{CD}
    \end{CD}
\end{equation}
we first need to explicate the tangent prolongation  Lie algebroid
$TB\to TM$. The definition is in~\cite[$\S$9.7]{mackenzie:book2005}.
We shall use the following proposition allowing to work directly with the homological vector fields.

\begin{proposition} \label{prop.tangprol} The tangent prolongation Lie algebroid $TE\to TM$ of a
Lie algebroid $E\to M$ specified by a homological vector field $Q\in \Vect (\Pi E)$ is given by the tangent prolongation vector field $\hat Q$, which is an (automatically homological) vector field  on $T(\Pi E)=\Pi_{TM}TE$.
\end{proposition}

{

Note that, for any vector bundle $E\to M$, taking tangents leads to a
double vector bundle
\begin{equation*}
    \begin{CD} TE@>>> E\\
                @VVV  @VVV \\
                TM@>>>M
    \end{CD}
\end{equation*}
(see~\cite[$\S$3.4]{mackenzie:book2005}).  In particular,    partial parity
reversions make sense.

Proposition~\ref{prop.tangprol}  is a completely natural statement and should be known to experts; however, when we needed it in~\cite{tv:mack}, we could not find it in the literature and had to check it ourselves. (In fact, the classical version\,---\,in terms of linear Poisson structures\,---\, was in~\cite[Theorem 5.6]{mackenzie:bialg}
\footnote{{ Compare with T.~Courant \cite[Theorem 6]{tcourant:tangent} where tangent Poisson structures were used for a  definition. This paper also contains    coordinate expressions for the tangent  Lie algebroid structure.}}.) For completeness, here we   include a proof, which can go as follows.

\begin{proof}[Proof of Proposition~\ref{prop.tangprol}]
First, we  express the bracket and the anchor for the tangent prolongation Lie algebroid in terms of a local basis of sections of $TE$ over $TM$. We shall use the notation of section~\ref{sec.algebroids}, so $e_i$ is a local frame for $E\to M$ and $u^i$ are the corresponding fiber coordinates. The induced  local frame for  $TE\to TM$ may be denoted $(\be_i, \bbe_i)$, so that $u^i\be_i+\dot u^i\bbe_i$ is invariant. In Mackenzie's notation,  $\bbe_i=\widehat{e_i}$\,;   these elements correspond to the basis $e_i$ of the core, which   in this case is the vector bundle $E\to M$ itself. At the same time, we have $\be_i=T(e_i)$, the tangents of the sections $e_i$. The definition of the bracket of sections of $TE\to TM$ given in~\cite[$\S$9.7]{mackenzie:book2005} translates, for the basis sections, into the equations
\begin{equation*}
    [\be_i,\be_j]=Q_{ij}^k\be_k\,+\,\dot Q_{ij}^k\bbe_k\,, \quad [\be_i,\bbe_j]=Q_{ij}^k\bbe_k
\,, \quad [\bbe_i,\bbe_j]=0\,.
\end{equation*}
Here, to avoid extra signs, the formulas are shown for the case when $E$ and $M$ are ordinary manifolds. There is no problem to write them in the general case. The  expressions such as $\dot Q_{i}^a$ stand for $\dot x^b \p_b Q_{i}^a$, etc. The definition of the anchor~\cite[$\S$9.7]{mackenzie:book2005} translates, similarly, into
\begin{equation*}
    a(\be_i)=Q_i^a\der{}{x^a}\,+\, \dot Q_i^a \der{}{\dot x^a}\,, \quad a(\bbe_i)= Q_i^a \der{}{\dot x^a}\,.
\end{equation*}
Now, we can compare this with the tangent prolongation of the   vector field on $\Pi E$\,,
\begin{equation*}
    Q=\x^iQ_i^a\,\der{}{x^a}+\frac{1}{2}\,\x^i\x^j Q_{ji}^k\,\der{}{\x^k}\,,
\end{equation*}
corresponding to the Lie algebroid structure of $E\to M$. The tangent prolongation $\hat Q$ is
\begin{multline*}
    \hat Q=\x^iQ_i^a\,\der{}{x^a}\,+\,\frac{1}{2}\,\x^i\x^j Q_{ji}^k\,\der{}{\x^k}   \\
    + \left(\dot\x^iQ_i^a + \x^i\dot Q_i^a \right)\,\der{}{\dot x^a} + \left(\dot \x^i\x^j Q_{ji}^k + \frac{1}{2}\,\x^i\x^j \dot Q_{ji}^k\right)\,\der{}{\dot \x^k}\,.
\end{multline*}
It is a vector field on $T(\Pi E)=\Pi_{TM} TE$. It is automatically homological because tangent prolongation maps the commutator of vector fields to the commutator of the prolongations.
One can  see immediately that the coefficients of this $\hat Q$ define precisely the brackets and the anchor written above, which completes the proof.  
\end{proof}

Coming back to the analysis of Condition \textrm{II}, we see that by differentiating the field $Q_{BM}$ given by~\eqref{eq.qbm}, we obtain a vector field $\hat Q_{BM}$ on $T(\Pi
B)=\Pi_{TM} TB$,
\begin{multline}\label{eq.qbmhat}
    \hat Q_{BM}=\h^{\a}Q_{\a}^a\der{}{x^a}
    + \frac{1}{2}\h^{\a}\h^{\b}Q_{\b\a}^{\g}\der{}{\h^{\g}}\\+
    \left(\dot\h^{\a}Q_{\a}^a+\h^{\a}\dot Q_{\a}^a\right)\der{}{\dot
    x^a}+
    \left(\dot\h^{\a}\h^{\b}Q_{\b\a}^{\g}+
    \frac{1}{2}\h^{\a}\h^{\b}\dot Q_{\b\a}^{\g}\right)\der{}{\dot\h^{\g}}
    \,,
\end{multline}
which is, by Proposition~\ref{prop.tangprol},   the homological vector field defining the tangent
prolongation Lie algebroid in~\eqref{eq.datotbtm}.

In order to simplify notation,   we shall continue writing all the  formulas in the rest of this section
for the case when  $D, A, B, M$ are ordinary manifolds, not
supermanifolds. This will allow us to avoid extra signs. Obviously,
everything carries over to the general super case.

}

Recall the formula for the anchor map $a\co D\to TB$.
From~\eqref{eq.qdb} we can obtain
\begin{equation}\label{eq.api}
    \dot x^a=u^iQ_i^a, \quad
    \dot \h^{\b}=u^i \h^{\a}Q_{\a i}^{\b}+ \t^{\mu}Q_{\mu}^{\b}
\end{equation}
for the corresponding map $\Pi_A D\to T(\Pi B)$.

We see that the condition that~\eqref{eq.datotbtm} is a morphism of
Lie algebroids translates into the condition that the map given
by~\eqref{eq.api} intertwines the vector fields $Q_{DA}$ and $\hat Q_{BM}$ given by~\eqref{eq.qda}
and~\eqref{eq.qbmhat}. After   simplification, this gives the
following four equations:
\begin{equation}\label{eq.anchor1}
    Q_{\mu}^{\b}\,Q_{\b}^a=Q_{\mu}^j\,Q_j^a\,,
\end{equation}
\begin{equation}\label{eq.anchor2}
    Q_{\a i}^{\b}\,Q_{\b}^a+Q_i^b\,\p_b
    Q_{\a}^a=Q_{\a}^b\,\p_bQ_i^a+Q_{i\a}^j\,Q_j^a\,,
\end{equation}
\begin{multline}\label{eq.anchor3}
    Q_{[\a
    i}^{\d}\,Q_{\b]\d}^{\g}+Q_i^b\,\p_bQ_{\b\a}^{\g}=
    Q_{[\a}^a\,\p_aQ_{\b]i}^{\g}+Q_{i[\a}^j\,Q_{\b]j}^{\g} +
Q_{\b\a}^{\d}\,Q_{\d i}^{\g}  +
Q_{i\b\a}^{{\lam}}\,Q_{{\lam}}^{\g}\,,
\end{multline}
and
\begin{equation}\label{eq.anchor4}
    Q_{\mu}^{\a}\,Q_{\b\a}^{\g}=-Q_{\b}^a\,\p_aQ_{\mu}^{\g}+Q_{\mu}^j\,Q_{\b
    j}^{\g}-Q_{\mu\b}^{{\lam}}\,Q_{{\lam}}^{\g}
\end{equation}
(the square brackets denote alternation in the appropriate indices, e.g., $Q_{[\a i}^{\d}\,Q_{\b]\d}^{\g}=Q_{\a i}^{\d}\,Q_{\b\d}^{\g}-Q_{\b i}^{\d}\,Q_{\a\d}^{\g}$).

To complete the analysis of Condition \textrm{II}, we have to consider
diagrams similar to~\eqref{eq.datotbtm1} and \eqref{eq.datotbtm}:
\begin{equation*}
    \begin{CD} D@>>> B\\
                @V{a}VV  @VV{a}V \\
                 TA@>>>TM
    \end{CD}
    \text{\quad and \quad}
    \begin{CD}
    \begin{CD} D@>>> B
    \end{CD}\\
    @VVV \\
    \begin{CD}  TA @>>>  TM
    \end{CD}
    \end{CD}
\end{equation*}
with $A$ and $B$ interchanged. This adds two equations to
the system~\eqref{eq.anchor1}--\eqref{eq.anchor4}:
\begin{multline}\label{eq.anchor5}
    Q_{[i\a}^l\,Q_{j]l}^k+ Q_{\a}^b\,\p_bQ_{ji}^k=Q_{[i}^a\,\p_aQ_{j]\a}^k+
    Q_{\a[i}^{\b}\,Q_{j]\b}^k+ Q_{ji}^l\,Q_{l\a}^k  +  Q_{\a
    ji}^{{\lam}}\,Q_{{\lam}}^k
\end{multline}
and
\begin{equation}\label{eq.anchor6}
    Q_{\mu}^i\,Q_{ji}^k=-Q_j^a\,\p_aQ_{\mu}^k+Q_{\mu}^{\b}\,Q_{j\b}^k-Q_{\mu
    j}^{{\lam}}\,Q_{{\lam}}^k\,.
\end{equation}
The system of equations~\eqref{eq.anchor1}--\eqref{eq.anchor6} is symmetric w.r.t. the exchange of $A$ and $B$. (Note that each of equations~\eqref{eq.anchor1} and \eqref{eq.anchor2} is symmetric w.r.t. this exchange. Equations~\eqref{eq.anchor3} and ~\eqref{eq.anchor4} are transformed to equations~\eqref{eq.anchor5} and \eqref{eq.anchor6}, respectively.)

\textit{The system of
equations~\eqref{eq.anchor1}--\eqref{eq.anchor6} is equivalent to
Condition {\em II} of Definition~\ref{def.dlie}.} Notice that it is
bilinear in the vector fields $Q_{DA}$ and $Q_{DB}$.

\subsection{Condition {\em III}, and  conclusion of the
proof}

Condition \textrm{III} requires ``deciphering'' more than the other conditions. It consists of two parts.
First of all, we have to consider the diagrams:
\begin{equation}\label{eq.dualda}
    \begin{CD} D^{*A}@>>> K^*\\
                @VVV  @VVV \\
                A@>>>M
    \end{CD}
\end{equation}
and
\begin{equation}\label{eq.dualdb}
    \begin{CD} D^{*B}@>>> B\\
                @VVV  @VVV \\
                K^*@>>>M
    \end{CD}
\end{equation}
and understand why the top horizontal arrow in~\eqref{eq.dualda} and
the left vertical arrow in~\eqref{eq.dualdb}  are   Lie algebroids with base $K^*$. This involves recalling some theory due to Mackenzie.

The \emph{core} $K$ of a double vector bundle~\eqref{eq.dvb2} (see~\cite[Ch.~9]{mackenzie:book2005}) is a vector bundle over the  base $M$ with fiber
coordinates   $z^{\mu}$ and the
transformation law
\begin{equation*}
   z^{\mu}=z^{\mu'}T_{\mu'}{}^{\mu}(x')
\end{equation*}
obtained from~\eqref{eq.lawforz} by setting $u^i$ and $w^{\a}$ to
zero.

When we dualize the vertical bundle $D\to A$ in~\eqref{eq.dvb},
we obtain the bundle $D^{*A}\to A$ with fiber coordinates $w_{\a},
z_{\mu}$ (with lower indices) so that the form
$w^{a}w_{\a}+z^{\mu}z_{\mu}$ giving the pairing in coordinates  is
invariant. We arrive at the following transformation laws:
\begin{align}
    w_{\a'}&=T_{\a'}^{\a} w_{\a}+u^{i'}T_{\a' i'}^{\mu}z_{\mu}\,,\\
    z_{\mu'}&=T_{\mu'}^{\mu}z_{\mu}\,,
\end{align}
where $(T_{\mu'}^{\mu})$ and $(T_{\mu}^{\mu'})$ are  reciprocal
matrices, and the transformation of $u^i$  remains as
in~\eqref{eq.lawforu}. This explains the double vector bundle
structure of~\eqref{eq.dualda}, in particular the vector bundle
$D^{*A}\to K^*$. (Note that $z_{\mu}$ can be identified with fiber
coordinates for $K^*\to M$.) The core of this new double vector bundle is $B^*\to M$.

The same holds when we dualize over $B$. The total space of the
vector bundle $D^{*B}\to B$ has coordinates $x^a, u_i, w^{\a},
z_{\mu}$, the coordinates $(u_i, z_{\mu})$ being dual to $(u^i,
z^{\mu})$ on $D$. Hence the transformation law
\begin{align}
    u_{i'}&=T_{i'}^{i} u_{i}+w^{\a'}T_{\a' i'}^{\mu}z_{\mu}\,,\\
    z_{\mu'}&=T_{\mu'}^{\mu}z_{\mu}\,,
\end{align}
from which we immediately obtain the double vector bundle structure
of~\eqref{eq.dualdb}. The new core is $A^*\to M$.

Treating $D^{*A}$ and $D^{*B}$ as vector
bundles over $K^*$, with fiber coordinates $(u^{i}, w_{\a})$ and
$(u_{i}, w^{\a})$, respectively, we arrive at a surprising natural
duality between them discovered in~\cite{mackenzie:sympldouble} (see
also~\cite{mackenzie:duality} and
\cite[$\S$9.2]{mackenzie:book2005}), with the pairing   given by the bilinear
form
\begin{equation}\label{eq.pairing}
     u^iu_i-w^{\a}w_{\a}\,.
\end{equation}
Here the minus sign between the two terms is absolutely essential for the invariance;
due to it the terms with $z_{\mu}$ appearing in the change of
coordinates cancel each other.

{
\begin{remark}
Unlike the minus sign between the two terms,  a choice of a common sign in  formula~\eqref{eq.pairing} cannot be fixed by  invariance considerations. The form  with the opposite sign $-u^iu_i+w^{\a}w_{\a}$ defines an equally good pairing of vector bundles. Therefore the `Mackenzie duality' between the  bundles $D^{*A}\to K^*$ and $D^{*B}\to K^*$ is canonical  up to a sign. This fact does not really affect anything in our constructions.
\end{remark}
}

Now, the Lie algebroid structures  on  the vector bundles $D^{*A}\to K^*$ and $D^{*B}\to
K^*$  follow as a consequence of the two non-trivial  facts:  1) the fiberwise linearity over
the base $K^*$ of the Poisson brackets  on the total spaces  $D^{*A}$ and
$D^{*B}$ induced by the Lie algebroid structures on the vector bundles $D\to A$ and $D\to B$,
respectively;  and  2) the above duality between the vector bundles
$D^{*A}\to K^*$ and $D^{*B}\to K^*$. By this duality, the  linear Poisson
structure on each bundle  $D^{*A}\to K^*$ and $D^{*B}\to K^*$ induces  a  Lie algebroid structure   on  the bundle  $D^{*B}\to K^*$  and $D^{*A}\to K^*$, respectively {(uniquely  up to a sign).}

This can be readily expressed in our language.

In order to obtain the homological vector fields specifying the Lie algebroid structures for $D^{*B}\to K^*$  and $D^{*A}\to K^*$,
one can   follow these steps: starting from the homological vector field  $Q_{DA}$ on $\Pi_AD$ construct the Lie-Poisson bracket on $D^{*A}$; use the linearity of this bracket over $K^*$ to obtain the Lie algebroid structure for the dual bundle $(D^{*A})^{*{K^*}}\to K^*$; apply the duality between $D^{*A}\to K^*$ and $D^{*B}\to K^*$ to re-write that as a Lie algebroid structure on $D^{*B}\to K^*$, and finally obtain the corresponding homological vector field on $\Pi_{K^*}D^{*B}$, which we denote $Q_{DA}^*$.  (Similarly, with $A$ and $B$ interchanged.) Alternatively, the following  shortcut   argument may be used. Since the homological vector field
$Q_{DA}$  on $\Pi_AD$ defining the vertical Lie algebroid structures of the original double vector bundle has weight $0$ in the horizontal direction, it generates fiberwise linear transformations  and therefore  allows the `transpose' in that direction. More precisely, one has to take the contragredient transformation (the adjoint of the inverse) for the corresponding linear flow and the generator of the new flow is the desired (automatically homological) vector field $Q_{DA}^*$ on   $\Pi_{K^*} D^{*B}$.

In coordinates $x^a, \x_i, \h^{\a}, z_{\mu}$ on $\Pi_{K^*} D^{*B}$, we obtain  the following expression for the homological vector field $Q_{DA}^*$:
\begin{multline}\label{eq.qonpikstardstara}
    Q_{DA}^*=\h^{\a}Q_{\a}^a\der{}{x^a} + \left(\x_j\h^{\a}Q_{i\a}^{j}-
    \frac{1}{2}\h^{\a}\h^{\b}z_{{\lam}}Q_{i \b\a}^{\lambda}\right)\der{}{\x_{i}} \\
   + \frac{1}{2}\h^{\a}\h^{\b}Q_{\b\a}^{\g}\der{}{\h^{\g}}
   + \left(\x_jQ_{\mu}^{j}-\h^{\a}z_{{\lam}}Q_{\mu\a}^{\lambda}\right)\der{}{z_{\mu}}\,.
\end{multline}
It defines the   Lie algebroid structure on $D^{*B}\to K^*\,$.

In the same way we calculate the homological vector field $Q_{DB}^*$ on $\Pi_{K^*} D^{*A}$ giving the  Lie algebroid structure on $D^{*A}\to K^*\,$. In coordinates $x^a, \x^i, \h_{\a}, z_{\mu}\,$, we have
\begin{multline}\label{eq.qonpikstardstarb}
    Q_{DB}^*=\x^{i}Q_{i}^a\der{}{x^a} +   \frac{1}{2}\x^{i}\x^{j}Q_{ji}^{k}\der{}{\x^{k}} \\
   +  \left(-\x^i\h_{\b}Q_{\a i}^{\b}- \frac{1}{2}\x^{i}\x^{j}z_{{\lam}}Q_{\a ji}^{\lambda}\right)\der{}{\h_{\a}}
   + \left(\h_{\b}Q_{\mu}^{\b}-\x^{i}z_{{\lam}}Q_{\mu i}^{\lambda}\right)\der{}{z_{\mu}}\,.
\end{multline}

One can immediately see that the vector field $Q_{DA}^*$ on $\Pi_{K^*} D^{*B}$  and $Q_{BM}$ on $\Pi B$ are related  by the
projection. Hence the horizontal
arrows in~\eqref{eq.dualdb} give a Lie algebroid morphism. The same is true for  the vector fields $Q_{DB}^*$ and  $Q_{AM}$, and for the vertical arrows in~\eqref{eq.dualda}.
\textit{We see that the first part of Condition {\em III} holds automatically!}

Let us examine the second part of Condition \textrm{III},   that \textbf{the dual
bundles $D^{*A}\to K^*$ and $D^{*B}\to K^*$  with the described Lie algebroid structures   form a Lie bialgebroid}.
\emph{It turns out to be the main condition.}

In our language it goes as follows. If we have two vector bundles in duality, a Lie algebroid structure on one of them induces the Lie-Schouten bracket on the total space of the parity-reversed other bundle, as we have recalled in \S\ref{sec.algebroids}. If both bundles are endowed with   Lie algebroid structures, the   condition that  they make a  Lie bialgebroid\footnote{More precisely,   two Lie bialgebroids in duality.} is equivalent to that for each parity-reversed  bundle the homological vector field is a derivation of the Schouten bracket (see~\cite{yvette:exact}, \cite[$\S$12.1]{mackenzie:book2005},  also~\cite{roytenberg:thesis} and \cite{tv:graded}). The condition is  self-dual and it suffices to check just one bundle. In our case, there are the homological vector fields, $Q_{DA}^*$    and $Q_{DB}^*$,  on $\Pi_{K^*} D^{*B}$ and $\Pi_{K^*} D^{*A}$ respectively,  and there are  the Schouten brackets induced on each bundle by duality.

Consider the bundle $\Pi_{K^*} D^{*B}$ for concreteness.
On $\Pi_{K^*} D^{*B}$ we obtain the following explicit expression for the Schouten
brackets between the coordinates:
\begin{equation}
\begin{gathered}\label{eq.schoutennnn}
    \{x^a,x^b\}=\{x^a,z_{\mu}\}=\{z_{\mu},z_{\nu}\}=0\,,\\
    \{\x_i, x^a\}=Q_i^a\,,\quad
    \{\x_i,z_{\mu}\}= -Q_{\mu i}^{{\lam}}z_{{\lam}}\,,\\
    \{\h^{\a},x^a\}=0\,, \quad
    \{\h^{\a},z_{\mu}\}=-Q_{\mu}^{\a}\,, \\
     \{\x_i,\x_j\}=Q_{ij}^k \x_k+Q_{\a ij}^{{\lam}}z_{{\lam}}\h^{\a}\,,\quad
    \{\x_i,\h^{\a}\}=Q_{\b i}^{\a}\h^{\b}\,,\\
    \{\h^{\a}, \h^{\b}\}=0\,.
\end{gathered}
\end{equation}
(In the same way we can find explicitly the Schouten bracket induced on $\Pi_{K^*} D^{*A}$\,, which we skip.)

The second part of Mackenzie's Condition \textrm{III} amounts therefore to the condition
that the vector field~\eqref{eq.qonpikstardstara}  is a derivation of the odd bracket~\eqref{eq.schoutennnn}.

{
\begin{remark} Since the duality between the vector bundles $D^{*A}\to K^*$ and $D^{*B}\to K^*$  is defined up to a sign, the Lie algebroid structures for them are also defined up to common signs, and so are the homological vector field~\eqref{eq.qonpikstardstara} and the Schouten bracket~\eqref{eq.schoutennnn}. Obviously, the  compatibility condition is  independent of these choices.
\end{remark}

By a direct calculation, the derivation property for the field~\eqref{eq.qonpikstardstara} w.r.t. the bracket~\eqref{eq.schoutennnn} expands    to the following system of nine equations:}
\begin{equation}\label{eq.bialg1}
    Q_{\a}^a\,Q_{\mu}^{\a}-Q_{\mu}^{i}\,Q_{i}^a=0
\end{equation}
\begin{equation}\label{eq.bialg2}
    -Q_{\mu}^{i}\,Q_{\nu i}^{{\lam}}+Q_{\mu}^{\a}\,Q_{\nu\a}^{{\lam}}-
    Q_{\nu}^i\,Q_{\mu i}^{{\lam}}+Q_{\nu}^{\a}\,Q_{\mu\a}^{{\lam}}=0
\end{equation}
\begin{equation}\label{eq.bialg3}
    Q_{\a
    j}^{\b}\,Q_{\b}^a+Q_j^b\,\p_bQ_{\a}^a-Q_{j\a}^i\,Q_i^a= Q_{\a}^b\,\p_bQ_j^a
\end{equation}
\begin{equation}\label{eq.bialg4}
     Q_{j}^{a}\,\p_aQ_{\mu}^i+Q_{\mu}^k\,Q_{jk}^i-Q_{\mu}^{\a}\,Q_{j\a}^i
    =-Q_{\mu j}^{{\lam}}\,Q_{{\lam}}^i
\end{equation}
\begin{multline}\label{eq.bialg5}
    Q_{\mu}^i\,Q_{\b ji}^{{\lam}}-Q_{\b j}^{\a}\,Q_{\mu\a}^{{\lam}}+Q_{\nu
    j}^{{\lam}}\,Q_{\mu \b}^{\nu}-Q_j^a\,\p_aQ_{\mu\b}^{{\lam}}+
    Q_{\mu i}^{{\lam}}\,Q_{j\b}^i-Q_{\mu}^{\a}\,Q_{j\b\a}^{{\lam}}
    \\=
    -Q_{\b}^a\,\p_aQ_{\mu j}^{{\lam}}+Q_{\mu j}^{\nu}\,Q_{\nu\b}^{{\lam}}
\end{multline}
\begin{equation}\label{eq.bialg6}
    -Q_{\mu}^i\,Q_{\a
    i}^{\g}+Q_{{\lam}}^{\g}\,Q_{\mu\a}^{{\lam}}-Q_{\mu}^{\b}\,Q_{\b\a}^{\g}=
    -Q_{\a}^a\,\p_aQ_{\mu}^{\g}
\end{equation}
\begin{multline}\label{eq.bialg7}
     Q_{\a j}^{\b}\,Q_{i\b}^k+Q_{j
    l}^k\,Q_{i\a}^l+Q_{j}^a\,\p_aQ_{i\a}^k-
    Q_{\a i}^{\b}\,Q_{j\b}^k-Q_{i
    l}^k\,Q_{j\a}^l-Q_{i}^a\,\p_aQ_{j\a}^k\\
    =
    Q_{\a}^a\,\p_aQ_{ij}^k-Q_{ij}^l\,Q_{l\a}^k-Q_{\a
    ij}^{\mu}\,Q_{\mu}^k
\end{multline}
\begin{multline}\label{eq.bialg8}
    Q_{\a ik}^{\lam}Q_{j\b}^k-Q_{\b ik}^{\lam}Q_{j\a}^k
- Q_{\a i}^{\g}Q_{j\b\g}^{\lam}
+Q_{\b i}^{\g}Q_{j\a\g}^{\lam}
+Q_{\mu i}^{\lam} Q_{j \b\a}^{\mu}- Q_i^a\p_aQ_{j\b\a}^{\lam}\\
-Q_{\a jk}^{\lam}Q_{i\b}^k+Q_{\b jk}^{\lam}Q_{i\a}^k
+ Q_{\a j}^{\g}Q_{i\b\g}^{\lam}
-Q_{\b j}^{\g}Q_{i\a\g}^{\lam}
-Q_{\mu j}^{\lam} Q_{i \b\a}^{\mu}+ Q_j^a\p_aQ_{i\b\a}^{\lam}
\\
= -Q_{ij}^kQ_{k\b\a}^{\lam} - Q_{\mu\a}^{\lam}Q_{\b ij}^{\mu}+ Q_{\mu \b}^{\lam}Q_{\a ij}^{\mu}\\+Q_{\a}^a\p_aQ_{\b ij}^{\lam}- Q_{\b}^a\p_aQ_{\a ij}^{\lam} + Q_{\g ij}^{\lam}Q_{\b\a}^{\g}\,,
\end{multline}
\begin{multline}\label{eq.bialg9}
    Q_{\b k}^{\g}\,Q_{j\a}^k
    -Q_{\a k}^{\g}\,Q_{j\b}^k
    +Q_{j\b\a}^{{\lam}}\,Q_{{\lam}}^{\g}
    -
    Q_{\a j}^{\e}\,Q_{\b\e}^{\g}
    +Q_{\b
    j}^{\e}\,Q_{\a\e}^{\g}-Q_j^a\,\p_aQ_{\b\a}^{\g}\\
    =
    -Q_{\b\a}^{\e}\,Q_{\e j}^{\g}-
    Q_{\a}^a\,\p_aQ_{\b j}^{\g}+Q_{\b}^a\,\p_aQ_{\a j}^{\g}
\end{multline}
Therefore, \textit{Condition {\em III} of Definition~\ref{def.dlie} is analytically expressed by the system of
equations~\eqref{eq.bialg1}--\eqref{eq.bialg9},  bilinear in the components of the original homological vector fields $Q_{DA}$ and $Q_{DB}$.}

Notice that equations~\eqref{eq.bialg1},  \eqref{eq.bialg3}, \eqref{eq.bialg4}, \eqref{eq.bialg6}, \eqref{eq.bialg7} and \eqref{eq.bialg9} in this system are already familiar. They are equivalent to equations~\eqref{eq.anchor1}, \eqref{eq.anchor2}, \eqref{eq.anchor6}, \eqref{eq.anchor4}, \eqref{eq.anchor5} and \eqref{eq.anchor3}, respectively, which together make an analytic expression of Condition \textrm{II}.  \textit{Hence,  Condition {\em III}  contains Condition {\em II}\,\footnote{After  this fact was  first discovered in~\cite{tv:mack}, Mackenzie    gave for it a different   proof within his original framework~\cite{mackenzie:double06}. See also~\cite{mackenzie:double11}.}}.  Since Condition \textrm{I} of the definition of a double Lie algebroid is encoded in the forms of the homological vector fields specifying the horizontal and vertical Lie algebroid structures, Condition \textrm{II} is subsumed by Condition \textrm{III}, and Condition \textrm{III} is expressed by the system of equations~\eqref{eq.bialg1}--\eqref{eq.bialg9} for the components of
these    fields, we conclude that  \textbf{Mackenzie's definition of a double Lie
algebroid reduces to this system of
equations~\eqref{eq.bialg1}--\eqref{eq.bialg9}}. This is the outcome of our analysis.

Note that this system is symmetric w.r.t.   swapping of $A$ and $B$ as it should. (In more detail, equations~\eqref{eq.bialg4} and \eqref{eq.bialg6} are exchanged under the transposition of $A$ and $B$,  and so are equations~\eqref{eq.bialg7} and \eqref{eq.bialg9}; each of the remaining five equations is symmetric itself.)

\smallskip
To finish the proof of Theorem~\ref{thm.main} it remains to compare
the system~\eqref{eq.bialg1}--\eqref{eq.bialg9} with the
commutativity condition for  the vector fields $Q_1=Q_{DB}^{\Pi}$ and $Q_2=Q_{DA}^{\Pi}$ on
$\Pi^2 D$ obtained from  $Q_{DB}$ and $Q_{DA}$ by partial  parity reversions:
\begin{multline}\label{eq.q1}
    Q_1=\x^iQ_i^a\der{}{x^a}+\frac{1}{2}\x^i\x^jQ_{ji}^k\der{}{\x^k} 
    +\left(\x^i\h^{\a}Q_{\a i}^{\b}+t^{\mu}Q_{\mu}^{\b}\right)\der{}{\h^{\b}}\ + \\
    \left(\frac{1}{2}\x^i\x^j\h^{\a}Q_{\a ji}^{\lambda}+\x^it^{\mu}Q_{\mu i}^{\lambda}\right)
    \der{}{t^{\lambda}}\,,
\end{multline}
and
\begin{multline}\label{eq.q2}
    Q_2=\h^{\a}Q_{\a}^a\der{}{x^a}+
    \left(\h^{\a}\x^iQ_{i\a}^{j}-t^{\mu}Q_{\mu}^{j}\right)\der{}{\x^{j}} 
    + \frac{1}{2}\h^{\a}\h^{\b}Q_{\b\a}^{\g}\der{}{\h^{\g}}\ + \\
    \left(-\frac{1}{2}\h^{\a}\h^{\b}\x^{i}Q_{i \b\a}^{\lambda}+\h^{\a}t^{\mu}Q_{\mu\a}^{\lambda}\right)
    \der{}{t^{\lambda}}\,.
\end{multline}
These formulas are written in the coordinates $x^a,\x^i,\h^{\a},t^{\mu}$ on $\Pi^2D$, see the discussion in \S\ref{sec.parreversion}. {Note that we use   the isomorphism $I_{12}$ (as defined in the proof of Proposition~\ref{prop.pp}) for the identification of $\Pi_A\Pi_BD$ and $\Pi_B\Pi_AD$.}
(Eqs.~\eqref{eq.q1} and \eqref{eq.q2} are similar to~\eqref{eq.fieldqa} and \eqref{eq.fieldqb},  with modified signs, see the remark after Eq.~\eqref{eq.qbm}. Note also that throughout this section we have worked  with a purely even $D$, for simplicity, so some particular signs have not shown up. All the calculations carry over to the general case of course.)

The calculation of the commutator of $Q_1$ and $Q_2$ is straightforward and shows that the
commutativity relation
$$[Q_1,Q_2]=0$$
expands to a system of equations that precisely coincides
with~\eqref{eq.bialg1}--\eqref{eq.bialg9}.

Hence we conclude that \emph{Mackenzie's Definition~\ref{def.dlie}
is equivalent to the commutativity of the homological fields $Q_1$
and $Q_2$}, \textsc{QUOD ERAT DEMONSTRANDUM}.

\section{The big picture}\label{sec.big}

After presenting a `computational' proof of our main statement, we
shall now give a conceptual explanation. The argument in this
section provides an alternative proof of Theorem~\ref{thm.main},
almost without calculations.

Let us again consider a double vector bundle
\begin{equation} \label{eq.dv1}
    \begin{CD} D@>>> B\\
                @VVV  @VVV \\
                A@>>>M
    \end{CD}
\end{equation}
We shall assume that all sides of it have Lie algebroid structures
and that each of these structures is compatible with the linear
structure  in the other direction. Thus we assume the obvious part
of the definition of a double Lie algebroid (i.e., Condition \textrm{I}). As
we have seen, this is equivalent to saying that the Lie algebroid
structures on $D\to A$ and $D\to B$ are defined by homological
vector fields of weights $(0,1)$ and $(1,0)$ on the ultimate total
spaces of
\begin{equation} \label{eq.dv2and3}
    \begin{CD} \Pi_AD@>>> \Pi B\\
                @VVV  @VVV \\
                A@>>>M
    \end{CD}
\text{\quad and \quad}
    \begin{CD} \Pi_BD@>>>  B\\
                @VVV  @VVV \\
                \Pi A@>>>M
    \end{CD}
\end{equation}
respectively.

We now show how a compatibility condition for these two Lie algebroid
structures can be introduced.

Let us return for a moment to ordinary Lie algebroids (or just Lie
algebras). Suppose $E\to M$ is a vector bundle. It has three
\textit{neighbors}: the dual bundle $E^*$, the opposite bundle $\Pi
E$ and the antidual $\Pi E^*$. A Lie algebroid structure in $E$
(which is a structure on the module of sections) is equivalently
expressed by each of the following structures on its neighbors: a
homological vector field of weight $1$ on $\Pi E$, a linear Poisson
bracket on $E^*$, and a linear Schouten bracket on $\Pi E^*$. The
axioms of a Lie algebroid are contained in the equation $Q^2=0$ or
in the Jacobi identities for the  Poisson or Schouten bracket. The
structures on $E^*$,  $\Pi E$ and  $\Pi E^*$ are   structures on the
total spaces{\,}\footnote{That is, on the algebras of functions as
opposed to a structure on  sections of a vector bundle or on
elements of a vector space.}.

Acting in a similar way, let us consider all the neighbors of our
double vector bundle~\eqref{eq.dv1}. There are four operations that
can be applied:  vertical dual,  horizontal dual, vertical reversion
of parity, and horizontal reversion of parity.
Besides~\eqref{eq.dv2and3} one   obtains the following double vector
bundles,  which are the neighbors of~\eqref{eq.dv1}.

The complete parity reversion of~\eqref{eq.dv1}:
\begin{equation} \label{eq.pikvadratt}
    \begin{CD} \Pi^2D@>>>  \Pi B\\
                @VVV  @VVV \\
                \Pi A@>>>M
    \end{CD}
\end{equation}
\\
The two duals of~\eqref{eq.dv1}:
\begin{equation} \label{eq.dv5and6}
    \begin{CD} D^{*A}@>>>  K^*\\
                @VVV  @VVV \\
                A@>>>M
    \end{CD}
\text{\quad and \quad}
    \begin{CD} D^{*B}@>>>  B\\
                @VVV  @VVV \\
                K^*@>>>M
    \end{CD}
\end{equation}
\\
The parity reversions of each of the duals:
\begin{equation} \label{eq.dv7}
    \begin{CD} \Pi_AD^{*A}@>>>  \Pi K^*\\
                @VVV  @VVV \\
                A@>>>M
    \end{CD} \qquad
    \begin{CD} \Pi_{B}D^{*B}@>>>   B\\
                @VVV  @VVV \\
               \Pi K^*@>>>M
    \end{CD}
\end{equation}
\\
\begin{equation} \label{eq.piduala}
    \begin{CD} \Pi_{K^*}D^{*A}@>>>   K^*\\
                @VVV  @VVV \\
                \Pi A@>>>M
    \end{CD} \qquad
    \begin{CD} \Pi_{K^*}D^{*B}@>>>  \Pi B\\
                @VVV  @VVV \\
                K^*@>>>M
    \end{CD}
\end{equation}
\\
and
\\
\begin{equation} \label{eq.pikvadratduala}
    \begin{CD} \Pi^2 D^{*A}@>>>   \Pi K^*\\
                @VVV  @VVV \\
                \Pi A@>>>M
    \end{CD} \qquad
    \begin{CD} \Pi^2 D^{*B}@>>>   \Pi B\\
                @VVV  @VVV \\
                \Pi K^*@>>>M
    \end{CD}
\end{equation}

This is the full list up to natural isomorphisms. These twelve
objects, including the original double vector bundle~\eqref{eq.dv1},
can be arranged into a four-valent colored  graph (where from each
vertex emanate two edges corresponding to taking the duals and two
edges corresponding to the parity reversions).

\begin{remark} In the multiple case, for an $n$-fold vector bundle, the
number of edges emanating from each vertex is $n+n=2n$.
\end{remark}

The Lie algebroid structures on the sides of ~\eqref{eq.dv1} obeying
the linearity conditions, which  were expressed above in terms of
weights, generate a pair of structures for each of the
neighbors~\eqref{eq.dv2and3}--\eqref{eq.pikvadratduala}, in various
combinations. More symmetrically, we may  say that each pair of
structures for a particular double vector bundle
from~\eqref{eq.dv1}--\eqref{eq.pikvadratduala} is just a
manifestation of one `pre-double structure' --- that is, without a
compatibility condition yet. (All pairs contain the same
information.) One can make a list of such structures. The next step
will be to look for suitable compatibility conditions for each pair.
The philosophy is that one should look for pairs where a
compatibility condition is formulated naturally, and take it as the
definition of compatibility for an equivalent pair where such a
condition does not come about in an obvious way.

In other words: suppose we do not know what  a double Lie algebroid
is; to get  the right notion of compatibility for the Lie algebroids
on the sides of the double vector bundle~\eqref{eq.dv1}, examine its
neighbors.

Our philosophy is that a verifiable compatibility condition lives on
the total space.

Let us start from~\eqref{eq.dv5and6}.  For each of the double vector
bundles in~\eqref{eq.dv5and6}, the ultimate total space is a Lie
algebroid (over the base $K^*$) and simultaneously possesses a
linear Poisson bracket (linear over both bases). That means that the
dual bundle over $K^*$ is also a Lie algebroid, and one may ask
whether they form a Lie bialgebroid. As the analysis of the previous
section shows, this may be considered as the \textbf{Mackenzie
definition} of a double Lie algebroid for~\eqref{eq.dv1}, since it
subsumes all his other conditions\,\footnote{As already noted, after  this fact was discovered~\cite{tv:mack},
Mackenzie  proved it directly in~\cite{mackenzie:double06}, see also~\cite{mackenzie:double11}, which  makes the analysis in this section entirely independent of our  coordinate calculations in Section~\ref{sec.analysis}.}.

Note that for Lie bialgebroids and Lie bialgebras, unlike   Lie
algebras or   Lie algebroids equivalently manifesting themselves on
a total space in terms of a linear Poisson bracket, or a linear
Schouten bracket, or a homological field of weight $1$,  everything
remarkably boils down to just one type of structure: namely, a
$QS$-structure, i.e., a homological vector field and a Schouten
bracket linked by the derivation condition (see~\cite{tv:graded}).

Now look at the neighbors~\eqref{eq.dv1}--\eqref{eq.pikvadratduala}.
Among them \emph{there are precisely five cases where a structure is
induced on the total space}:~\eqref{eq.pikvadratt}, and the four
bundles in~\eqref{eq.piduala}, \eqref{eq.pikvadratduala}.

On the total space  of each of the double vector bundles
in~\eqref{eq.piduala} there is a Schouten bracket of weight
$(-1,-1)$ and a homological vector field of weight $(0,1)$ or
$(1,0)$, respectively. On the total space of the bundles
in~\eqref{eq.pikvadratduala} there is a Poisson bracket  of weight
$(-1,-1)$ and a homological vector field of weight $(0,1)$ or
$(1,0)$. A compatibility condition in each case is the derivation
property of the vector field w.r.t. the bracket.

\begin{proposition} \label{prop.bialg}
The compatibility conditions for the four bundles
in~\eqref{eq.piduala} and \eqref{eq.pikvadratduala} are equivalent,
and are but different ways of saying that $(D^{*A}, D^{*B})$ is a Lie
bialgebroid over $K^*$.
\end{proposition}
\begin{proof} Indeed, for any Lie bialgebroid $(E,E^*)$ the compatibility can
be stated in terms of either $E$ or $E^*$ (as a $QS$-structure on
either $\Pi E$ or $\Pi E^*$, respectively). This corresponds to one
of the bundles in~\eqref{eq.piduala}. In our special situation there
is also an extra option of changing parity in the other direction
(and turning a $QS$-structure into a $QP$-structure), which adds two
equivalent descriptions in terms of the bundles
in~\eqref{eq.pikvadratduala}.
\end{proof}

The remaining case is  the total space of~\eqref{eq.pikvadratt}
where there are two homological vector fields of weights $(0,1)$ and
$(1,0)$. A compatibility condition for them is of course
commutativity.

We see now that there are essentially two conditions to compare: the
Mackenzie bialgebroid condition, which lives on one of the bundles
in~\eqref{eq.piduala}, \eqref{eq.pikvadratduala}, and our
commutativity relation for the bundle~\eqref{eq.pikvadratt}.

\begin{proposition} \label{prop.equiv}
The Mackenzie bialgebroid condition and the commutativity of the two
homological vector fields  for~\eqref{eq.pikvadratt} are equivalent.
\end{proposition}
\begin{proof} Consider one of the manifestations of the bialgebroid
condition, say, for concreteness, on the first double vector bundle
in~\eqref{eq.piduala}. The derivation property means that the flow
of the vector field preserves the Schouten bracket. On the other
hand, the commutativity condition for~\eqref{eq.pikvadratt} means
that the flow of one field preserves the other. Now the claim
follows by functoriality:  a linear transformation preserves a Lie
bracket if and only if the adjoint or anti-adjoint map preserves the
corresponding linear Poisson or Schouten bracket and if and only if
the `$\Pi$-symmetric' map preserves the corresponding homological
vector field.
\end{proof}

Propositions~\ref{prop.bialg} and~\ref{prop.equiv} together imply
Theorem~\ref{thm.main}.

\section{Applications and generalizations}\label{sec.appl}

In this section we revise `Drinfeld doubles' for Lie bialgebroids,
introduce systematically multiple Lie algebroids already touched
upon in the previous sections, and discuss how Lie bialgebroid
theory may be  extended to the multiple case.

\subsection{Doubles of Lie bialgebroids}

We already mentioned in the introduction that the problem of a
Drinfeld double of a Lie bialgebroid was one of the motivations and
a testing case for Mackenzie's definition of double Lie algebroids.

Recall that Drinfeld's \textit{classical double} of a Lie bialgebra
is again a Lie bialgebra with ``good'' properties. An analog of this
construction for Lie bialgebroids turned out to be a puzzle.  Three
constructions of a `double' have been suggested. Suppose $(E,E^*)$
is a Lie bialgebroid over a base $M$. Liu, Weinstein and
Xu~\cite{weinstein:liuxu} suggested to consider as its double a
structure of a Courant algebroid on the direct sum $E\oplus E^*$.
Mackenzie in~\cite{mackenzie:doublealg}, \cite{mackenzie:drinfeld},
\cite{mackenzie:notions} and Roytenberg in~\cite{roytenberg:thesis}
suggested two different constructions based on the cotangent bundles
$T^*E$ and $T^*\Pi E$, respectively.

Roytenberg showed~\cite{roytenberg:thesis} that the
Liu--Weinstein--Xu double can be recovered from his own construction
using derived brackets, by generalizing the results of
C.~Roger~\cite{roger:1991} and
Y.~Kosmann-Schwarzbach~\cite{yvette:jacobian}, \cite{yvette:derived}
for Lie bialgebras.  Thus only the constructions
of~\cite{roytenberg:thesis} and~\cite{mackenzie:doublealg},
\cite{mackenzie:drinfeld}  have to be compared.

Though approaches of~\cite{roytenberg:thesis}
and~\cite{mackenzie:doublealg}, \cite{mackenzie:drinfeld} look very
different, we shall  now establish their equivalence.

Both Roytenberg's and Mackenzie's construction use the statement
that the cotangent bundles of dual vector bundles are
isomorphic~\cite{mackenzie:bialg}. There is a double vector bundle
\begin{equation} \label{eq.mackdouble}
    \begin{CD} T^*E=T^*E^*@>>>   E^*\\
                @VVV  @VVV \\
               E@>>>M
    \end{CD}
\end{equation}
Mackenzie~\cite{mackenzie:doublealg} shows that it is a double Lie
algebroid, which he calls the `cotangent double' of a Lie
bialgebroid $(E,E^*)$. He uses his original definition of double Lie
algebroids and we do not need to elaborate his argument here.

On the other hand, Roytenberg~\cite{roytenberg:thesis}  considers
the diagram
\begin{equation} \label{eq.roytdouble}
    \begin{CD} T^*\Pi E=T^*\Pi E^*@>>>   \Pi E^*\\
                @VVV  @VVV \\
               \Pi E@>>>M
    \end{CD}
\end{equation}
His further construction is as follows. Suppose $Q_E\in\Vect (\Pi
E)$ and $Q_{E^*}\in\Vect (\Pi E^*)$ are the homological vector
fields defining the Lie algebroid structures on $E\to M$ and $E^*\to
M$, respectively. Assign to them the fiberwise linear functions
$H_E$ and $H_{E^*}$ on the cotangent bundles $T^*\Pi E$ and $T^*\Pi
E^*$, respectively. One proves~\cite{roytenberg:thesis} that under
the natural symplectomorphism $T^*\Pi E \to T^*\Pi E^*$ the linear
function $H_{E^*}$ on $T^*\Pi E^*$ corresponding to the vector field
$Q_{E^*}$ is transformed into the fiberwise quadratic function $S_E$
on $T^*\Pi E$ specifying the Schouten bracket on $\Pi E$ induced by
the Lie structure on $E^*$. Therefore the derivation property of
$Q_E$ w.r.t. the Schouten bracket on $\Pi E$, which  is the most
convenient definition  of a Lie bialgebroid~\cite{yvette:exact} is
equivalent to the commutativity of the Hamiltonians $H_E$ and
$H_{E^*}$ under the canonical Poisson bracket. They generate
commuting homological vector fields $X_{H_E}$ and $X_{H_{E^*}}$ on
the cotangent bundle $T^*\Pi E$. It was suggested
in~\cite{roytenberg:thesis} to consider the sum
$Q=X_{H_E}+X_{H_{E^*}}$, which is a homological field of total
weight $+1$ on the graded manifold $T^*\Pi E$ as the desired
`double'.

If we bear in mind that $T^*\Pi E$ is a double vector bundle and
check that $X_{H_E}$ and $X_{H_{E^*}}$ have  the right weights
$(1,0)$ and $(0,1)$, we can conclude that this construction leads
exactly to a double Lie antialgebroid. To compare it with
Mackenzie's construction we need to apply Theorem~\ref{thm.main}.

Notice that $\Pi^2 T^*E$ coincides with $T^* \Pi E$. Hence the
double vector bundle
\begin{equation} \label{eq.mackdoublepikv}
    \begin{CD} \Pi^2T^*E=\Pi^2T^*E^*@>>>   \Pi E^*\\
                @VVV  @VVV \\
               \Pi E@>>>M
    \end{CD}
\end{equation}
obtained by the complete parity reversion of the double vector
bundle~\eqref{eq.mackdouble} is identical
with~\eqref{eq.roytdouble}. It remains to identify the respective
homological vector fields on the total space, which is achieved by a
direct inspection. We arrive at the following statement.
\begin{proposition} Roytenberg's and Mackenzie's pictures give the same notion of a double of a Lie
bialgebroid.
\end{proposition}

We can now identify the two constructions and speak simply of the
\textit{(cotangent) double} of a Lie bialgebroid.  The Proposition shows that the cotangent double is fundamental and should be regarded
as the correct extension of Drinfeld's double to Lie bialgebroids.

\begin{remark}
The canonical symplectic structure on  $T^*\Pi E$ preserved by the
Hamiltonian homological vector fields $X_{H_E}$ and $X_{H_{E^*}}$
corresponds to the invariant scalar product on Drinfeld's double
$\mathfrak d (\mathfrak b)=\mathfrak b\oplus \mathfrak b^*$ of a Lie
bialgebra $\mathfrak b$.
\end{remark}

\subsection{Multiple Lie algebroids}

We have already mentioned that the methods of this paper allow us to consider a natural
concept of multiple Lie algebroids. (A part of the
motivation for doing that comes from the theory of doubles.) We
shall give an outline of this theory.

First we need a language for describing multiple vector bundles. Fix
a natural number $n$. To define $n$-fold vector bundles, consider
vector spaces $V_r$, $V_{r_1r_2}$, $V_{r_1r_2r_3}$, \ldots, of
arbitrary dimensions $d_r$,$d_{r_1r_2}$, $d_{r_1r_2r_3}$, etc.,
numbered by increasing sequences $r_1<\ldots <r_k$, where $0<k\leq
n$ and all $r_i$ run from $1$ to $n$.

\begin{example} When $n=1$, we have just one vector space $V=V_1$. When
$n=2$, we have $V_1$, $V_2$ and $V_{12}$. For $n=3$, we have $7$
spaces: $V_1$, $V_2$, $V_3$, $V_{12}$, $V_{13}$, $V_{23}$, and
$V_{123}$. In general the number of spaces is $2^n-1$.
\end{example}

For convenience of notation let us fix linear coordinates on each of
the spaces, denoting them $v^{i_r}_{(r)}$,
$v^{i_{r_1r_2}}_{(r_1r_2)}$, etc. (Each index such as $i_r$ runs
over its own set of values, of cardinality equal to the dimension of
the respective space.)

\begin{definition} An \textbf{$n$-fold vector bundle} $E$ over a base $M$ is a fiber
bundle   $E\to M$ with the standard fiber
\begin{equation*}
    \prod_{r} V_r \,\times \,\prod_{r_1<r_2} V_{r_1r_2} \,\times \,\ldots
    \,\times\,
    V_{12\ldots n}
\end{equation*}
where the transition functions have the form:
\begin{align*}
    v^{i_r}_{(r)}&=v^{i_r'}_{(r)} T_{i_r'}^{i_r},  \\
    v^{i_{r_1r_2}}_{(r_1r_2)}&=v^{i_{r_1r_2}'}_{(r_1r_2)} T_{i_{r_1r_2}'}^{i_{r_1r_2}}+
    v^{i_{r_1}'}_{(r_1)}v^{i_{r_2}'}_{(r_2)}
    T_{i_{r_2}'i_{r_1}'}^{i_{r_1r_2}},
     \intertext{\quad\quad\quad\quad\quad\quad\quad\dots\dots\dots\dots}
    v^{i_{12\ldots n}}_{(12\ldots n)}&=v^{i_{12\ldots n}'}_{(12\ldots
    n)}T_{i_{12\ldots n}'}^{i_{12\ldots n}}+\ldots + v^{i_1'}_{(1)}\ldots
    v^{i_n'}_{(n)} T_{{i_n'}\ldots{i_1'}}^{i_{12\ldots n}}\,.
\end{align*}
\end{definition}

In other words, the transformation for each of $V_r$ is linear; for
$V_{r_1r_2}$ it is linear plus an extra term bilinear in $V_{r_1}$
and $V_{r_2}$, etc.

\begin{example} For a triple vector bundle $E\to M$ ($n=3$), we have fiber coordinates:
$v^{i_1}_{(1)}$, $v^{i_2}_{(2)}$, $v^{i_3}_{(3)}$,
$v^{i_{12}}_{(12)}$, $v^{i_{13}}_{(13)}$, $v^{i_{23}}_{(23)}$, and
$v^{i_{123}}_{(123)}$. The transformation law is as follows:
\begin{align*}
    v^{i_1}_{(1)}&=v^{i_1'}_{(1)}T_{i_1'}^{i_1}\\
    v^{i_2}_{(2)}&=v^{i_2'}_{(2)}T_{i_2'}^{i_2}\\
    v^{i_3}_{(3)}&=v^{i_3'}_{(3)}T_{i_3'}^{i_3}\\
    v^{i_{12}}_{(12)}&=v^{i_{12}'}_{(12)}T_{i_{12}'}^{i_{12}}+
    v^{i_1'}_{(1)}v^{i_2'}_{(2)}T_{{i_2'}{i_1'}}^{i_{12}}\\
    v^{i_{13}}_{(13)}&=v^{i_{13}'}_{(13)}T_{i_{13}'}^{i_{13}}+
    v^{i_1'}_{(1)}v^{i_3'}_{(3)}T_{{i_3'}{i_1'}}^{i_{13}}\\
    v^{i_{23}}_{(23)}&=v^{i_{23}'}_{(23)}T_{i_{23}'}^{i_{23}}+
    v^{i_2'}_{(2)}v^{i_3'}_{(3)}T_{{i_3'}{i_2'}}^{i_{23}}\\
    v^{i_{123}}_{(123)}&=
                    \begin{aligned}[t]
    v^{i_{123}'}_{(123)}T_{i_{123}'}^{i_{123}}+
    v^{i_1'}_{(1)}v^{i_{23}'}_{(23)}T_{{i_{23}'}{i_1'}}^{i_{123}}+
    v^{i_2'}_{(2)}v^{i_{13}'}_{(13)}T_{{i_{13}'}{i_2'}}^{i_{123}}+
    v^{i_3'}_{(3)}v^{i_{12}'}_{(12)}T_{{i_{12}'}{i_3'}}^{i_{123}} \\+
v^{i_1'}_{(1)}v^{i_2'}_{(2)}v^{i_3'}_{(3)}T_{{i_3'}{i_2'}{i_1'}}^{i_{123}}
\end{aligned}
\end{align*}
\end{example}

\begin{remark} Triple vector bundles\,---\,with the quaternary case
briefly mentioned\,---\,were introduced and studied
in~\cite{mackenzie:duality} from a different viewpoint (not using
local trivializations and transition functions).
Paper~\cite{mackenzie:duality} also contains some `likely
principles' of duality for general multiple case. The existence of a local trivialization was not discussed in~\cite{mackenzie:duality}; under a form of a ``decomposition'' it was explicitly introduced in the definition in~\cite{mackenzie:duality2009}. To avoid any problems and since it is not our task to minimize axiomatic systems, here we define multiple vector bundles as a particular case of locally trivial fiber bundles from the start.
\end{remark}

A multiple vector bundle has \textit{faces}, which are also multiple
vector bundles. A face is obtained by choosing indices $r_1<\ldots <
r_k$; fiber coordinates for it will be the coordinates
$v^{i_{{r_1\ldots r_k}}}_{(r_1\ldots r_k)}$ and all other
coordinates with indices labelled by subsets of $r_1,\ldots, r_k$.
For example, for a triple vector bundle there are faces that are
(ordinary) vector bundles and double vector bundles, corresponding
to the edges and $2$-faces of a $3$-cube. In a natural way  various
partial projections and zero sections are defined.

The total space of a multiple vector bundle is a multi-graded
manifold. More precisely, there are weights $\boldsymbol{w}_r$,
$r=1,\ldots, n$, each of them being a degree  in all coordinates
containing a given label $r$. For example, $\boldsymbol{w}_2$ is the
total degree in $v_{(2)}$, $v_{(12)}$, $v_{(23)}$, \ldots,
$v_{(12\ldots n)}$. We define \textit{total weight} as
$\boldsymbol{w}=\boldsymbol{w}_1+\ldots+\boldsymbol{w}_n$.

Due to the multilinearity of transition functions, for a multiple
vector bundle  the operations of \textit{partial parity reversion}
$\Pi_r$ and   \textit{partial dual} ${D}_r$ in the $r$-th
direction, make sense for each $r=1,\ldots, n$.

\begin{definition} \label{def.multalie}
An \textbf{$n$-fold Lie antialgebroid} $E$ over a base $M$ is an $n$-fold
vector bundle $E\to M$ endowed with $n$ odd vector fields $Q_r$ of
weights $(0,\ldots, 1, \ldots, 0)$ on the total space $E$ such that
\begin{equation*}
    [Q_r,Q_s]=0
\end{equation*}
for all $r,s$. (In particular, these fields are homological.)
\end{definition}

\begin{definition} \label{def.multlie}
An \textbf{$n$-fold Lie algebroid} $E$ over a base $M$  is an $n$-fold
vector bundle $E\to M$  such that the $n$-fold vector bundle $\Pi^n
E\to M$ obtained by the complete parity reversion $\Pi^n=\Pi_n\ldots
\Pi_1$ is an $n$-fold Lie antialgebroid.
\end{definition}

In other words, we take the statement of Theorem~\ref{thm.main} as a
working definition for the multiple case.

Each face of a multiple Lie (anti)algebroid is also a multiple Lie
(anti)algebroid.

We expect that it is possible to define multiple Lie algebroids also
\textit{\`{a}  la} Mackenzie, via duals and bialgebroids, and to
show the equivalence  with Definition~\ref{def.multlie} (i.e., to
prove the analog of Theorem~\ref{thm.main}). This will require an
analysis of the structures induced on the neighbors of a multiple
Lie (anti)algebroid.

\subsection{More on doubles}

Recall that Drinfeld's classical double of a Lie bialgebra is not
just a Lie algebra, but also a coalgebra, and furthermore a Lie
bialgebra again. This gives a direction in which to look in the case
of Lie bialgebroids. Note that this second structure (for doubles of
Lie bialgebroids) has not been discovered previously.

There is a conjectured statement that reads as follows.

\begin{genprinc}
Taking the double of an $n$-fold Lie bialgebroid  gives  an
$(n+1)$-fold Lie bialgebroid, with   additional properties such as a
symplectic structure.
\end{genprinc}

Of course it involves new notions yet to be defined.   Multiple Lie
algebroids were introduced above. As for ``double Lie
\textbf{bi}algebroids'' (or ``\textbf{bi}- double Lie algebroids''),
and, further,  the ``bi-''   multiple case,   this is a subject of
our forthcoming joint paper with Kirill Mackenzie. The example and
discussion  below should be seen just as preliminary hints.

\begin{example} \label{ex.double}
Consider again the double vector bundle given
by~\eqref{eq.mackdouble}. It is a double Lie algebroid. Notice that
the core of it is the cotangent bundle $T^*M\to M$. Take the two
duals of the double vector bundle~\eqref{eq.mackdouble}. We obtain
the double vector bundles
\begin{equation} \label{eq.mackdouble3}
    \begin{CD} TE@>>>   T M\\
                @VVV  @VVV \\
               E@>>>M
    \end{CD}
\end{equation}
as the vertical dual and
\begin{equation} \label{eq.mackdouble4}
    \begin{CD} TE^*@>>>   E^*\\
                @VVV  @VVV \\
               T M@>>>M
    \end{CD}
\end{equation}
as the horizontal dual. Both~\eqref{eq.mackdouble3} and
\eqref{eq.mackdouble4} are known to be double Lie algebroids as
well~\cite{mackenzie:notions}. The three double vector bundles in
duality~\eqref{eq.mackdouble}, ~\eqref{eq.mackdouble3},
~\eqref{eq.mackdouble4}
\begin{equation}
    \begin{picture}(300,80)(0,65)
    \put(90,90){${\begin{CD} {TE}  @>>>   T M\\
                @VVV  @VVV \\
               E@>>>M\end{CD}}$}
    \put(120,110){$\begin{CD} {\hphantom{T^*T^*E}} @.   \!\!T E^*\\
                @.  @VVV \\
               T^*E@>>>\!\!E^*\end{CD}$}
    {\put(132,84){\vector(-2,-1){18}}}  
    {\put(192,84){\vector(-2,-1){18}}}  
    {\put(192,126){\vector(-2,-1){18}}} 
    \end{picture}
\end{equation}
(we use the picture of a corner~\cite{mackenzie:duality} for them),
which are all double Lie algebroids, can be  said to be making a
`bi- double Lie algebroid'. One can mean by that, for example, that
the triple vector bundle
\begin{equation}
    \begin{picture}(300,80)(0,65)
    \put(90,90){${\begin{CD} {TE}  @>>>   T M\\
                @VVV  @VVV \\
               E@>>>M\end{CD}}$}
    \put(120,110){$\begin{CD} T^*T^*E @>>>   \!\!T E^*\\
                @VVV  @VVV \\
               T^*E@>>>\!\!E^*\end{CD}$}
    {\put(132,84){\vector(-2,-1){18}}}
        {\put(132,126){\vector(-2,-1){18}}}
    {\put(192,84){\vector(-2,-1){18}}}
    {\put(192,126){\vector(-2,-1){18}}} 
   \end{picture}
\end{equation}
is a triple Lie algebroid, as one can check.
\end{example}

To define a `bi- multiple Lie algebroid'  one can also use the
supergeometry language for a shortcut. A  \textbf{bi- $n$-fold Lie
algebroid}   (or an \textbf{$n$-fold Lie bialgebroid}) can be
defined as an $n$-fold Lie algebroid $E$ such that all its duals are
also $n$-fold Lie algebroids satisfying the compatibility condition
that reads as follows: on the total space with completely reversed
parity $\Pi^nE$ there are $n$ commuting homological vector fields
$Q_r$ of weights $\w_r(Q_s)=\delta_{rs}$, which define  the
algebroid structures, and an odd or even (depending on the parity of
the number $n$) Poisson bracket of weight $(-1,\ldots,-1)$, and the
fields $Q_r$ are derivations of the bracket.

It should be possible to prove that this is equivalent to the
cotangent double being an $(n+1)$-fold Lie algebroid, as in
Example~\ref{ex.double}. A precise relation of this multidimensional
notion with Drinfeld's theory is yet to be clarified.

\bigskip
{

\small
\emph{Acknowledgement. }
I thank my friends Kirill Mackenzie, Yvette Kosmann-Schwarzbach and
Hovhannes Khudaverdian, for inspiring discussions, most valuable
criticism, and advice. K.~Mackenzie has pioneered the whole subject
of multiple and bi- structures in the groupoid and algebroid world.
I thank most cordially Y.~Kosmann-Schwarzbach for numerous comments
and remarks that helped to improve the original version
of the preprint~\cite{tv:mack} and were used in this text as well. My special
thanks go to the organizers of the annual international Workshops on
Geometric Methods in Physics in Bia{\l}owie\.{z}a, notably to
A.~Odzijewicz, for the highly inspiring atmosphere. Some of my early
notes on the subject of this paper were made in Warsaw after the
XXII Bia{\l}owie\.{z}a Workshop, and the main result was reported at
the XXV Workshop. The earlier version of the present paper became
known to experts as the preprint~\cite{tv:mack} and has already
influenced some works. It was also reported at a special program at
the Erwin Schroedinger Institute in Vienna in September 2007
(see~\cite{tv:qman-esi}). I thank the organizers of that program for
their hospitality and the wonderful working atmosphere. I also wish to thank the anonymous referees for the careful reading of the text and important remarks.

}


\def\cprime{$'$} \def\cprime{$'$} \def\cprime{$'$}

\end{document}